\documentclass[12pt]{article}

\usepackage{newpxtext}

\usepackage[margin=1in]{geometry}
\usepackage{amsmath,amssymb,amsthm,amsbsy,amscd,mathtools}
\usepackage{thmtools}
\usepackage{bm}
\usepackage{bbm}
\usepackage{latexsym}
\allowdisplaybreaks
\numberwithin{equation}{section}
\usepackage{textcomp}
\usepackage{microtype}
\usepackage{framed}
\usepackage{comment}
\usepackage{booktabs}

\usepackage{tikz}
\usetikzlibrary{arrows.meta,decorations.pathreplacing,positioning,calc}

\usepackage[most]{tcolorbox}
\tcbset{breakable}
\usepackage{graphicx}
\usepackage{float}
\usepackage[width=.85\textwidth]{caption}
\usepackage{subcaption}
\usepackage{diagbox}
\usepackage{makecell}
\usepackage{arydshln} 
\usepackage{blkarray}
\usepackage{multirow}
\usepackage{rotating}
\usepackage{enumitem}
\usepackage{algorithmic}
\usepackage[ruled,noend]{algorithm2e}
\usepackage{tikz}
\usetikzlibrary{
  calc,
  positioning,
  fit,
  shapes,
  shapes.misc,
  shapes.geometric,
  decorations.pathmorphing,
  decorations.markings,
  arrows.meta,
  backgrounds,
  patterns
}
\usepackage[backend=bibtex,sorting=nyt,maxbibnames=99]{biblatex}
\usepackage{breakcites}
\usepackage{hyperref}
\hypersetup{
  colorlinks=true,
  citecolor=blue,
  linkcolor=blue
}
\usepackage{fancyhdr}
\setlist[itemize]{
  topsep=4pt,
  itemsep=2pt,
  parsep=0pt,
  leftmargin=1.2em,
  label=\textbullet
}
\setlist[enumerate]{
  topsep=4pt,
  itemsep=2pt,
  parsep=0pt,
  leftmargin=1.5em,
  label=\textbf{\arabic*.}
}

\normalsize
\def\abs#1{\left\lvert#1\right\rvert}

\def\Abs#1{\bigl\lvert#1\bigr\rvert}

\newcommand{\R}{\mathbb{R}}

\newcommand{\eps}{\varepsilon}
\newcommand{\ep}{\varepsilon}

\newcommand{\defeq}{:=}

\newcommand{\PNE}{{\mathbf{PNE}}}
\newcommand{\Po}{\bm{\mathrm{Poi}}}

\newcommand{\lp}{\left(}
\newcommand{\rp}{\right)}

\newcommand{\cC}{{\calC}}

\newcommand{\calE}{\mathcal E}
\newcommand{\calF}{\mathcal F}
\newcommand{\calC}{\mathcal C}

\newcommand{\calD}{\mathcal D_n}
\newcommand{\calP}{\mathcal P}

\newcommand{\calI}{\mathcal I}
\newcommand{\calJ}{\mathcal J}
\newcommand{\calT}{\mathcal T}

\newcommand{\calL}{\mathcal L}
\newcommand{\calR}{\mathcal R}

\newcommand{\calU}{\mathcal U}

\newcommand{\red}[1]{{\color{red}{#1}}}
\newcommand{\blue}[1]{{\color{blue}{#1}}}

\newcommand{\rui}[1]{{\bf\color{red}{#1}}}

\tikzset{
  onethirdarrow/.style={
    decoration={
      markings,
      mark= at position 0.4 with {\arrow{#1}}
    },
    postaction={decorate}
  }
}
\tikzset{
  twothirdarrow/.style={
    decoration={
      markings,
      mark= at position 0.75 with {\arrow{#1}}
    },
    postaction={decorate}
  }
}

\renewcommand{\ge}{\geqslant}
\renewcommand{\le}{\leqslant}
\renewcommand{\geq}{\geqslant}
\renewcommand{\leq}{\leqslant}
\newcommand{\E}[1]{\mathbb E\left[\,#1\,\right]}

\newcommand{\V}[1]{\mathbf{Var}\left[\,#1\,\right]}
\renewcommand{\P}[1]{\mathbb{P}\left(#1\right)}
\newcommand{\Exp}[1]{\exp\left(#1\right)}
\renewcommand{\O}[1]{O\left(#1\right)}

\newcommand{\Log}[1]{\log\left(#1\right)}

\newcommand{\floor}[1]{\left\lfloor\,#1\,\right\rfloor}
\newcommand{\ceil}[1]{\left\lceil\,#1\,\right\rceil}

\renewcommand{\r}{\mathbf r}
\renewcommand{\c}{\mathbf c}

\renewcommand{\i}{\mathbf i}
\renewcommand{\j}{\mathbf j}
\newcommand{\I}[1]{\mathbbm 1_{\left\{\,#1\,\,\right\}}}
\renewcommand{\I}[1]{\textbf{1}_{\{#1\}}}

\newcommand{\bR}[1]{\left(#1\right)}
\newcommand{\bC}[1]{\Big\{#1\Big\}}
\newcommand{\bS}[1]{\left[ #1 \right]}
\renewcommand{\dfrac}[2]{\lower0.12ex\hbox{\large$\textstyle\frac{#1}{#2}$}}
\newtheorem{theorem}{Theorem}[section]
\newtheorem{tm}[theorem]{Theorem}
\newtheorem{lemma}[theorem]{Lemma}

\newtheorem{claim}[theorem]{Claim}
\newtheorem{proposition}[theorem]{Proposition}

\newtheorem{definition}[theorem]{Definition}

\newtheorem{conj}[theorem]{Conjecture}

\usepackage[T1]{fontenc}
\renewcommand{\,}{\hspace{0.05em}}

\title{\LARGE{\textbf{Cycle structure of best-response dynamics in random bimatrix games}}}
\author{
\normalsize
Georgios Nikolaos Karelas\thanks{
Department of Economics and Business, 
Pompeu Fabra University; 
Barcelona School of Economics,\\
georgios.karelas@upf.edu.}
\and
\normalsize
G\'abor Lugosi\thanks{
Department of Economics and Business, 
Pompeu  Fabra University;
ICREA, Pg. Lluis Companys 23, 
08010 Barcelona, Spain; 
Barcelona School of Economics,
gabor.lugosi@gmail.com.}
\and
\normalsize
Rui-Ray Zhang\thanks{
Monash University,
rui.ray.zhang@hotmail.com.}
}
\date{}
\begin{document}
\maketitle

\begin{abstract}
We study the best-response dynamics in games with random payoffs. We consider two-player bimatrix games, with $n$ strategies each. We assume that payoffs come from a continuous distribution, and study the best-response dynamics of these games, after imposing some structure on their payoff matrices. In particular, we focus on symmetric games, antisymmetric games, common-interest games, zero-sum games, and games with independent payoffs. When two players take turns in playing the best response to the opponent's action, they eventually end up cycling through the same sequence of actions. The main focus of the paper is the cycle structure of large random games. In particular, we study the number of cycles, the number of cycles of a certain length, the length of the longest cycle, and related quantities. We show that the structure of the game has an important impact on the typical cycle structure. Our arguments are based on a representation of the best-response dynamics of these games as directed bipartite graphs, and on the cycle structure that arises in them.
\end{abstract}

\section{Introduction}
Best-response dynamics is arguably the most natural way of repeatedly playing a game in multi-agent systems. This dynamical process starts from an initial strategy profile, and at each step, players sequentially update their strategies by choosing the pure strategy that yields the highest payoff, given the other player's chosen strategy. We study best-response dynamics within the framework of two-player, $n$-action random games, assuming payoffs are drawn from the same continuous distribution. 

Because best-response dynamics frequently lead to indefinite cycling over a subset of strategies rather than settling at a pure Nash equilibrium, a deeper investigation into this non-converging behavior is of particular interest. 
These cycles, together with the pure Nash equilibria, act as the natural ``sink equilibria'' of the game. 

In this paper, we provide a rigorous analysis of the typical cycle structure of random games. We consider five types of random games, defined by the relation between the payoff matrices of both players. These games include: games with independent payoffs, symmetric, antisymmetric, zero-sum, and common-interest games. Using the properties of the games, we find the expected total number of cycles, the number of short cycles, as well as the length of the longest cycle and other related combinatorial quantities that characterize the typical geometry of the non-convergent behavior of the best-response dynamics.

\subsection{Random payoffs
and the best-response digraphs
}
\label{section:types of games}
Let $\bm A$ and $\bm B$ denote the two $n\times n$ payoff matrices of the two players. 
The game structure is determined by the relationship between these two matrices. In particular, we consider five types of games, each imposing a different structure on the players' payoff matrices.


The first type is the random game with independent payoffs, in which no structure is imposed. More precisely, the two $n\times n$ matrices $\bm A$ and $\bm B$ of the two players have i.i.d. entries, and they are independent of each other.  We also consider symmetric games, where the payoff matrix of each player is the transpose of their opponent's. More precisely, the matrix $ \bm A$ of the first player has i.i.d.\ entries, and $ \bm B= \bm A^{\mathsf{T}}$. The third type of game is an antisymmetric game, where $\bm B=-\bm A^{\mathsf{T}}$. We also study common-interest games, where players have the same payoff matrices, $\bm A=\bm B$. Lastly, we consider zero-sum games, in which the players' payoffs satisfy $\bm B=-\bm A$.

Note that
for all cases,
we first assume matrix  $\bm A$ has i.i.d. entries,
and then $\bm B$ is defined according to the game types,
the five classes are mutually exclusive almost surely,
as each entry follows the i.i.d. continuous distribution.

\begin{figure}[htb]
\centering
\begin{tikzpicture}[
    font=\sffamily\small,
    node distance=2.1cm and 2.2cm,
    box/.style={
        draw,
        rounded corners=4pt,
        thick,
        fill=white,
        inner sep=4pt,
        minimum width=3.3cm,
        minimum height=1.15cm,
        align=center,
        font=\sffamily\small
    },
    arrow/.style={-{Stealth[length=2.6mm,width=1.8mm]}, line width=1.2pt},
    labeltext/.style={font=\sffamily\footnotesize\itshape}
]
\node[box] (ind)
  {{\bfseries Independent}\\[1pt]{\footnotesize $\bm A$ and $\bm B$ are independent}};
\node[box, above=of ind] (sym)
  {{\bfseries Symmetric}\\[1pt]{\footnotesize $\bm B=\bm A^{\mathsf T}$}};
\node[box, below=of ind] (asym)
  {{\bfseries Antisymmetric}\\[1pt]{\footnotesize $\bm B=-\bm A^{\mathsf T}$}};
\node[box, left=of ind] (zero)
  {{\bfseries Zero-sum}\\[1pt]{\footnotesize $\bm B=-\bm A$}};
\node[box, right=of ind] (common)
  {{\bfseries Common-interest}\\[1pt]{\footnotesize $\bm B=\bm A$}};
\draw[arrow] (ind) -- (sym)
    node[midway, right=2pt, labeltext] {Symmetry};
\draw[arrow] (ind) -- (asym)
    node[midway, right=2pt, labeltext] {Anti-symmetry};
\draw[arrow] (ind) -- (zero)
    node[midway, above=1pt, labeltext] {Conflict};
\draw[arrow] (ind) -- (common)
    node[midway, above=1pt, labeltext] {Cooperation};
\end{tikzpicture}
\label{fig:types-of-games}
\caption{Relations among random games of different types,
which are exclusive almost surely.}
\end{figure}
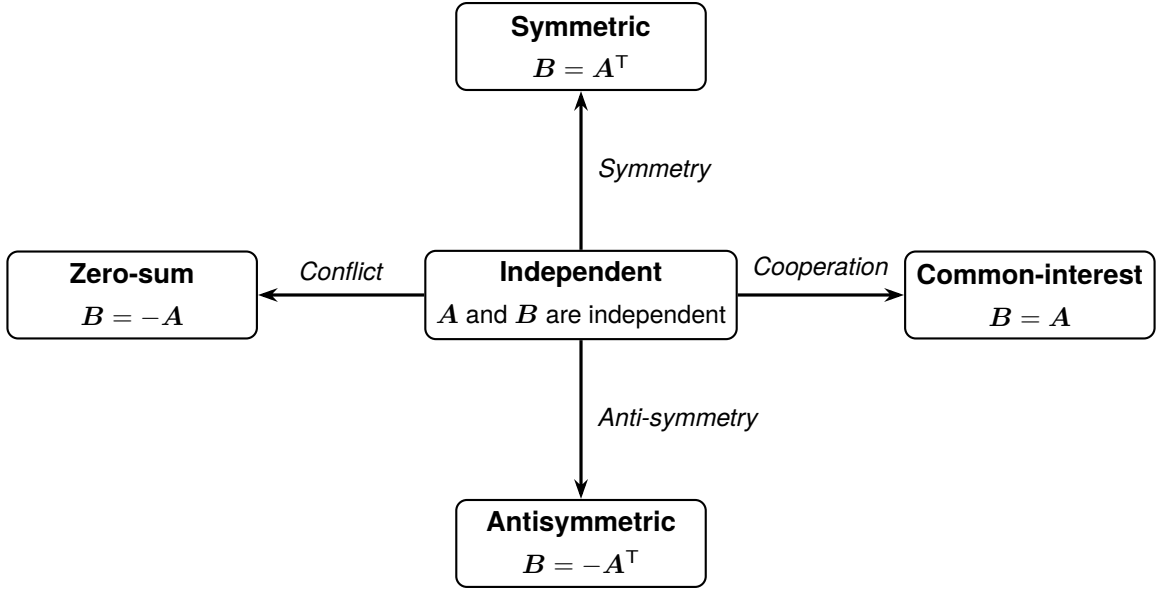


By a \emph{digraph} (or directed graph), we mean
$D = (V,E)$ where $V(D)$ is a nonempty set of distinct elements called
\emph{vertices}; and $E(D)$ is a set of ordered pairs $(u,v)$ of distinct
vertices $u,v \in V$ called \emph{arcs}.
We use a natural representation of the best-response dynamics in a two-player game by bipartite graphs, which we call {\it best-response digraphs}.
Best-response digraphs allow us to extract all quantities of interest, such as the number of cycles, their length, etc.
We use $[n]$ for the set of integers from $1$ up to $n$, for any integer $n\ge1$. 
The set of all $k$-element subsets of the set $[n]$ is denoted by $\binom{[n]}{k}$.

\begin{tcolorbox}
\begin{definition}
Given a bimatrix game with $n\times n$ payoff matrices $\bm A = (a_{i,j})$ and $\bm B=(b_{i,j})$, 
we define a directed bipartite graph $\calD = \calD( \bm A, \bm B )$ as follows. $\calD$ has 
a vertex bipartition, $R$ and $C$,
where $R = \{r_1, \ldots, r_n\}$
and $C = \{c_1, \ldots, c_n\}$.
We call $R$ and $C$ the row and column parts, respectively. 
Digraph $\calD$ contains arc $(r_i, c_j)$ if and only if $ b_{i,j} = \max_{k \in [n]}b_{i,k}$;
and similarly,
$(c_j, r_i) \subset \calD$ if and only if $ a_{i,j} = \max_{k \in [n]}a_{k,j}$\,. 
\end{definition}
\end{tcolorbox}

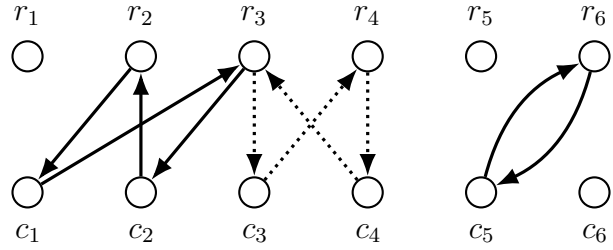
\begin{figure}[htbp]
\centering
\begin{tikzpicture}[x=1.8cm,y=1.5cm,>=Latex,rotate=90]
  \tikzset{
    row/.style={ellipse,draw,thick,minimum width=0.1cm,minimum height=0.1cm},
    col/.style={ellipse,draw,thick,minimum width=0.1cm,minimum height=0.1cm},
    rededge/.style={->,very thick},
    blueedge/.style={->,very thick,dotted},
    greenedge/.style={->,very thick,bend left=15}
  }
  \foreach \i in {1,...,6}{
    \node[row] (r\i) at (0,-\i) {};
    \node[col] (c\i) at (1,-\i) {};
    \node[anchor=east] at (-0.3,-\i-0.22) {$c_{\i}$};
    \node[anchor=west] at (1.3,-\i+0.22) {$r_{\i}$};
  }
  \draw[rededge] (r1) -- (c3);
  \draw[rededge] (c3) -- (r2);
  \draw[rededge] (r2) -- (c2);
  \draw[rededge] (c2) -- (r1);
  \draw[blueedge] (r3) -- (c4);
  \draw[blueedge] (c4) -- (r4);
  \draw[blueedge] (r4) -- (c3);
  \draw[blueedge] (c3) -- (r3);
  \draw[greenedge] (r5) to[bend left=25] (c6);
  \draw[greenedge] (c6) to[bend left=25] (r5);
\end{tikzpicture}
\caption{An example of the best-response digraph. 
 The dotted cycle cannot appear in $\calD$, otherwise, strategy $r_3$ has 2 best responses, namely strategies $c_2$ and $c_3$. The pair of strategies $(r_6,c_5)$ is a pure Nash equilibrium.}
\label{fig:enter-label}
\end{figure}

The two sets of $n$ vertices correspond to the $n$ strategies available to each player. 
Two vertices (one from each set) are connected with a directed edge whenever the strategy at the head of the directed edge is a best response to the strategy at the tail of that edge. If a pair of strategies $\{r_i, c_j\}$ 
induces a directed $2$-cycle in $\calD$, 
then $(r_i, c_j)$ is a pure Nash equilibrium.

\begin{tcolorbox}[breakable]
    \begin{definition}
    Let $K_{n,n}$ be the complete bipartite directed graph with bipartition $(R,C)$, where $R=\{r_{1},\ldots,r_{n}\}$ and $C=\{c_{1},\ldots,c_{n}\}$.
For a fixed $\ell\ge1$, consider size-$\ell$ index sets $\calI=\{i_1<i_2<\ldots <i_\ell\}\in \binom{[n]}{\ell}$ and $\mathcal  J=\{j_1 < j_2 < \ldots < j_\ell\}\in
\binom{[n]}{\ell}$. 
Define $S_m$ as the set of permutations on $[m]$. The set of directed $2\ell$-cycles with vertex sets $\calI,\calJ$ is defined as
    \begin{align*}
        &\calC_{2\ell}(\calI,\calJ) 
= \Big\{ (r_{i_1}, c_{j_{\tau(1)}}, r_{i_{\sigma(2)}}, c_{j_{\tau(2)}}, \dots, r_{i_{\sigma(\ell)}}, c_{j_{\tau(\ell)}})
        ~~\big|~~\sigma\in S_{\ell-1},\tau\in S_\ell\Big\},
    \end{align*}
    where each element of $\calC_{2\ell}(\calI,\calJ)$  corresponds to the directed closed walk
    
    $$r_{i_1}\to c_{j_{\tau(1)}}\to  r_{i_{\sigma(2)}}\to  c_{j_{\tau(2)}}\to  \cdots\to  r_{i_{\sigma(\ell)}}\to  c_{j_{\tau(\ell)}}\to r_{i_1}\,.$$
    
    The set of all directed $2\ell$-cycles in $K_{n,n}$ is 
    $$\calC_{2\ell}
    =\bigcup_{\calI, \calJ \in \binom{[n]}{\ell}}\calC_{2\ell}(\calI,\calJ)\,.$$
For a cycle $C=(r_{i_1},c_{j_1},\ldots,r_{i_\ell},c_{j_\ell})\in\calC_{2\ell}$, the position
subscript $k$ in $i_k$ and $j_k$ is always modulo $\ell$; in particular $i_{\ell+1}= i_1$.    
For $C \in \cC_{2\ell}$, let $\calI(C) \subseteq [n]$ be the set of row indices
used by the left vertices of $C$, and let $\calJ(C) \subseteq [n]$ be the set of
column indices used by the right vertices of $C$. 
The \emph{index overlap} of $C$ is
\[
  Z(C) \defeq |\calI(C) \cap \calJ(C)|\,.
\]
    \end{definition}
\end{tcolorbox}
It is easy to check (see also, e.g., \cite[Eq. (1)]{pangallo2019best}) that the number of directed $2\ell$-cycles in $K_{n,n}$ is
\begin{align}
  \Abs{ \calC_{2\ell} }
      = \binom{n}{\ell}^2
\Abs{\calC_{2\ell}([\ell],[\ell])}
    = \binom{n}{\ell}^2
\frac{\ell!^2}{\ell}
    = \frac{(n)_\ell^2}{\ell}.
\label{cycleCount}
\end{align}
Since the payoffs follow a continuous distribution, ties in the payoffs occur with probability zero.
Hence, there is exactly one best response given the other player's strategy, and therefore, there is exactly one directed edge coming out of every node on both sides of the best-response digraph. 
This observation leads to the next lemma about the structural properties of the best-response digraphs, whose proof is a simple exercise.

\begin{tcolorbox}
\begin{lemma}
\label{vertex-disjoint lemma}
Consider a bimatrix game, 
where the best response is unique
given the other player's strategy.
Let $\calD$ be the corresponding best-response digraph.
\begin{enumerate}
\item The digraph $\calD$ is a 
directed $1$-pseudoforest,
that is, a directed graph, in which each vertex has outdegree exactly one;
\item Any two distinct cycles in the graph are vertex-disjoint in $\calD$;
\item 
If two paths in $\calD$ intersect, then from their first point of intersection onward, they coincide.
\item 
If the game is one of the five random models of Section~\ref{section:types of games},
then for any vertex $u$, 
the in-degree of $u$ converges, in distribution, to 
$\Po(1)$.
\end{enumerate}
\end{lemma}
\end{tcolorbox}

\subsection{Overview of results and related work}
\label{related work}

The study of Nash equilibria and best-response dynamics in random games is a central research topic in game theory.
In tables \ref{tab:pne} and \ref{tab:cycles}, we give a summary of the results of the paper, together with some existing results in the literature.

\begin{table}[H]
\centering
\small
\renewcommand{\arraystretch}{2}
\setlength{\aboverulesep}{0pt}
\setlength{\belowrulesep}{0pt}
\setlength{\tabcolsep}{6pt}

\begin{tabular}{@{}l|cccc}
\toprule
Games & Payoffs & $\E{N_\PNE}$ & $N_\PNE$ limit law  & Reference\\
\midrule
Independent     & $\bm A,\bm B$ independent  & $1$                           & $\Po(1)$                                                 & \cite{powers1990limiting}\\
Symmetric       & $\bm B=\bm A^{\mathsf T}$  & $2-\frac1n$                   & $2\,\Po(1/2)+\Po'(1)$                                  & \cite{stanford1996limit}\\
Antisymmetric   & $\bm B=-\bm A^{\mathsf T}$ & $1-\frac1n$                   & $\Po(1)$                                                 & Theorem \ref{anti-cycle}\\
Zero-sum        & $\bm B=-\bm A$             & $\frac{(n!)^2}{(2n-1)!}$      & $\P{N_\PNE=0}\to1$                                       & \cite{Goldman1957ThePO}\\
Common-interest & $\bm B=\bm A$              & $\frac{n^2}{2n-1}$            & $ (2N_\PNE-n) \sqrt{\frac{2}{n}} \to
\mathcal N(0,1)$ & \cite{stanford1999number,baldi1989normal}\\
\bottomrule
\end{tabular}
\caption{Pure Nash equilibria correspond to $2$-cycles in $\calD$.}
\label{tab:pne}
\end{table}

\begin{table}[H]
\renewcommand{\arraystretch}{2}
\setlength{\aboverulesep}{0pt}
\setlength{\belowrulesep}{0pt}
\centering
\begin{tabular}{@{}l|ccc@{}}
\toprule
 & \makecell{Independent/ \\ 
 Antisymmetric/\\ Zero-sum} & Symmetric & Common-interest\\
\midrule
$\E{X_\ell}$
 & $\frac1\ell e^{-\ell^2/n}$
 & $\frac1\ell e^{-2\ell^2/n}+\frac1\ell e^{-\ell^2/(2n)}\I{\ell\text{ odd}}$
 & $0$\\
$X_\ell$ limit law
 & $\Po(1/\ell)$
  & $2\,\Po(1/(2\ell))+\Po'(1/\ell)\I{\ell\text{ odd}}$
 & $0$\\
$\E{\text{total cycle count}}$
 & $\frac12\log n$
 & $\frac34\log n$
 & $\frac{n^2}{2n-1}$\\
Longest cycle $2\ell_{\max}$
 & $\Theta_{\mathbb P}(\sqrt n)$
 & $\Theta_{\mathbb P}(\sqrt n)$
 & $2$\\
\bottomrule
\end{tabular}
\caption{The variable $X_{\ell}$ counts the best-response $2\ell$-cycles in $\calD$
for $\ell\ge2$. The first row shows the asymptotic values of $\E{X_\ell}$ that are valid when $\ell=o(n^{2/3})$. The third row is the asymptotic expected total cycle count.}
\label{tab:cycles}
\end{table}

\noindent
\textbf{Independent payoff games.} Foundational work in this area includes Goldberg, Goldman and Newman \cite{goldberg1968probability}, who studied random two-player games with independent payoffs and found the probability of existence of a pure-strategy Nash equilibrium, as well as Dresher \cite{dresher1970probability}, who extended this result to multiplayer games. Powers \cite{powers1990limiting} and Stanford \cite{stanford1995note}, considered multiplayer random games and demonstrated that the number of pure-strategy Nash equilibria follows approximately a Poisson distribution with mean $1$. Rinott and Scarsini \cite{rinott2000number} extended this analysis by considering large random games with many players, providing asymptotic results for the distribution of pure Nash equilibria in games with independent payoffs. More recently, Amiet, Collevecchio, Scarsini and Zhong \cite{amiet2021pure} studied multiplayer random games with 2 strategies each, finding asymptotic results for the number of pure Nash equilibria and showing that best-response dynamics converges to a pure Nash equilibrium with high probability when the number of players grows and the probability of ties is small. Lastly, McLennan \cite{mclennan2005expected} studied the number of mixed-strategy Nash equilibria in random games with $n$ players and provides a formula for their mean. 
On the algorithmic side, Collevecchio, Lugosi, Vetta and Zhang
\cite{winlose} study random bimatrix games with i.i.d.\ Bernoulli
payoffs and use a similar bipartite digraph representation to identify
combinatorial structures supporting mixed Nash equilibria.
\medskip \\
\noindent
\textbf{Symmetric games.} Stanford \cite{stanford1996limit} analyzed the asymptotic distribution of pure Nash equilibria in symmetric random games. The symmetric structure introduces dependencies that significantly affect the equilibrium properties compared to the independent case.
\medskip \\
\noindent
\textbf{Zero-sum games.} Goldman \cite{Goldman1957ThePO} was among the first to study the probability of existence of pure Nash equilibria in random zero-sum games. Rinott and Scarsini \cite{rinott2000number} also covered zero-sum games in their comprehensive analysis of different random game structures. 
\medskip \\
\noindent
\textbf{Random mappings.}
The best-response digraphs considered here are functional digraphs: every vertex has exactly one outgoing arc.
Their cycle structure is therefore naturally related to the classical
theory of random mappings
\cite{flajolet1989random,harris1960probability,purdom1968cycle,kolchin1986random}.
For independent payoffs, the best-response digraph is a random bipartite mapping,
and Jaworski \cite{jaworski1985random} studied this model, obtaining exact distributions and asymptotic results, including the
number of connected components and cyclic points.
We give self-contained proofs of the cycle-count and cycle-length
results.
\medskip \\
\noindent
\textbf{Common-interest games.} Stanford \cite{stanford1999number} studied the number of pure-strategy Nash equilibria in random finite common payoff games. This structure leads to fundamentally different dynamics, as players' interests are aligned, resulting in best response graphs with sink cycles of length 2,
which correspond to pure Nash equilibria. Common-interest games are also known as identical interest games,
which are a type of exact potential games
in the literature.
They
are special cases of potential games,
which have been studied extensively,
and are known to have 
a pure Nash equilibrium,
which can be obtained through the best-response dynamics. Another example of potential games is 2-player random ordinal potential games, where players have the same number of strategies, for which Collevecchio, Mimun, Quattropani, and Scarsini \cite{collevecchio2024basins} find the size of the basins of attraction of the pure Nash equilibria of the game after running the best-response dynamics.
\medskip \\
\noindent
\textbf{General correlated payoff structures.} 
Recent work has begun to explore games with more general correlation structures. Mimun, Quattropani, and Scarsini \cite{mimun2024best} considered 2-player random games where players' payoffs are equal with probability $p$ and independent and identically distributed with probability $1-p$. They studied the number of pure Nash equilibria and showed that best-response dynamics converges to a pure Nash equilibrium with high probability for any non-zero $p$. 
\medskip \\
\noindent
\textbf{Convergence and dynamics.} Beyond counting the number of equilibria, several works have focused on the convergence properties of best-response dynamics. Heinrich \textit{et al.} \cite{heinrich2023best} studied the two-player setting, computing the probability that best-response dynamics started from a random profile converges to a $2k$-cycle and the convergence time. Further, they showed that in large games, clockwork playing order rarely reaches a pure Nash equilibrium even if one exists.
Amiet \textit{et al.} \cite{amiet2021pure} studied the convergence of best-response dynamics to pure Nash equilibria in random games with many players and strategies, examining how the order of play affects convergence. They proved that when the playing sequence is random, best-response dynamics reaches a pure Nash equilibrium (if it exists) with high probability as the number of players grows to infinity. Pangallo, Heinrich, and Farmer \cite{pangallo2019best} took a different approach, 
studying several boundedly rational learning dynamics (reinforcement learning, 
fictitious play, replicator dynamics, experience-weighted attraction, and level-$k$ 
reasoning) in random two-player games whose payoffs are correlated across players, ranging from zero-sum to common-interest games. Exploring this space numerically, they showed that the best-reply cycles of a game predict whether these dynamics converge. 
Amiet, Collevecchio, and Hamza \cite{amiet2021better} prove that for i.i.d. payoffs, with high probability, the better-response dynamics reaches a pure Nash equilibrium
(if one exists), while the best-response dynamics fails to do so. Further, Johnston \textit{et al.} \cite{johnston2023game} show that in all but an exponentially small fraction of games that have a pure Nash equilibrium and have more players than the maximum number of strategies $k$ available for each one, every action profile that is not a pure Nash equilibrium can reach every pure Nash equilibrium by the best-response dynamics; while in their later work \cite{johnston2026game} they also study the large-$k$ regime. Swenson, Murray, and Kar \cite{swenson2018best} studied the best-response dynamics in potential games, showing that it converges to a pure Nash equilibrium at an exponential rate for almost every game and every initial condition. Regarding a different game dynamic, Biggar and Papadimitriou \cite{biggar2025sink} study the long-run outcomes of the replicator dynamic by exploiting their relationship to strongly connected components of a game’s preference
graph.  Lastly, Wiese and Heinrich \cite{wiese2022frequency} use an $n$-partite graph to describe games with $n$ players and $m$ strategies, and after generating payoff matrices at random, show that under the best-response dynamics, the fraction of games that converge to a pure Nash equilibrium goes to zero as $n$ or $m$ goes to infinity.

Our work contributes to this literature by analyzing different game types, independent, symmetric, antisymmetric, common-interest, and zero-sum games, and by studying the cycle structure of the resulting best-response digraphs.

\subsection{Probabilistic tools and asymptotic notation}

Throughout the paper, we use the following notation.
We write BR($r_i$) = $c_j$ if strategy $c_j$ of the column player is the best response to strategy $r_i$ of the row player, 
and BR($c_i$) = $r_j$ 
is similarly defined.
The $k$-th falling factorial of $x$ is denoted by
$(x)_k \defeq x(x-1)(x-2)\cdots(x-k+1)$ for $k\ge1$.
Further, we write $f(n)\sim g(n)$ if $\lim_{n\to \infty}(f(n)/g(n))=1$, and use $\I{\cdot}$ to denote the indicator function. 
We write
$
f_n=\Theta_{\mathbb P}(g_n)
$
if for every $\varepsilon>0$, there exist constants
$0<c_{\ep}<C_{\ep}<\infty$ such that
\[
\P{ c_{\ep} g_n \le f_n \le C_{\ep} g_n } 
\ge 1-\varepsilon
\]
for all sufficiently large $n$.

We make use of a version of the first and second moment methods 
to bound the probability that certain substructures do not appear in $\calD$.

\begin{tcolorbox}
[title=The first and second moment methods]
\begin{lemma}
\label{2rd-moment}
Let  
$\{ Y_t: t \in \mathcal T\, \}$ be a finite set of random indicators, and define
\[
Z = \sum_{t \in \mathcal T} Y_t,
\quad
\mu=\sum_{t\in\calT} \E{Y_t},
\quad\mbox{and}\quad
\Delta
=\sum_{ s\ne t\in\calT}
\E{ Y_s Y_t }.
\]
For $\mu>0$, if
$\Delta
\le \alpha \mu+(1+\beta) \mu^2$
for some $\alpha, \beta\ge0$, then
\begin{align}
\P{ Z \ge 1 } \le \mu,
\quad \text{and} \quad
\P{ Z = 0 } 
\le \frac{1+\alpha}
{\mu}+\beta.
\label{probX=0}
\end{align}
Moreover, if $\calT \subseteq
\cup_{\ell \in [n]}\calC_{2\ell}$
and let $Y_C = \I{ C \subset \calD }$ for $C \in \calT$, then
\[
\Delta
=
\sum_{
\substack {C_i\ne C_j\in\calT \\
C_i \cap C_j=\emptyset}
}
\E{ Y_{C_i}Y_{C_j} }.
\]
\end{lemma}
\end{tcolorbox}
Note that the summations above are over ordered pairs in the definition of $\Delta$,
and for cycles, it is over vertex-disjoint pairs. The proof is a simple application of the first and second moment methods, where the last part follows
from Lemma \ref{vertex-disjoint lemma},
as two cycles must be vertex-disjoint to appear together in $\calD$.

In view of \eqref{cycleCount},
it is useful to define, for $1\le\ell\le n$, 
\begin{align}
    \mu_{\ell}
    =\mu_{\ell}(n)
    \defeq 
    \sum_{C \in \calC_{2\ell}} n^{-2\ell}
    =\frac{1}{\ell} \prod_{i=1}^{\ell-1}
    \bR{1-\frac{i}{n}}^2,
\label{muDef}
\mbox{ and }
\quad
    \hat\mu_{\ell}
    =\hat\mu_{\ell}(n)
    \defeq \frac{1}{\ell}e^{-\ell^2/n}.
\end{align}

The following is a simple observation:

\begin{lemma}
\label{lem:muEquiv}
If $\ell = o(n^{2/3})$, then
$\mu_{\ell} \sim \hat\mu_{\ell}$ as $n\to\infty$.
Moreover, for all $2 \le \ell \le n$,
\begin{align}
    \mu_{\ell}
    \le
    \frac{1}{\ell}\,
    \Exp{-\frac{(\ell-1)^2}{n}}.
    \label{mubound}
\end{align}
\end{lemma}

\section{Random independent games}
In this section, we study 
random games with independent payoffs, where both payoff matrices $\bm A,\bm B$ have i.i.d. entries and $\bm A$ and $\bm B$ are independent. 
We calculate the expected number of cycles of a fixed length and their asymptotic distribution, the expected total number of cycles, and the typical length of the longest cycle.
\subsection{Cycle count}
We first determine the expected value and asymptotic distribution of the short-cycle count in the best-response digraph.
\begin{tcolorbox}
\begin{tm}
\label{2lcyclesInd}
Let $\calD$ be the  
best-response digraph
of an
$n \times n$
random independent bimatrix game.
Let $X_\ell$ count the number of $2\ell$-cycles in $\calD$.
Then for any integer $\ell\in [n]$,
$$
\E{X_\ell}= \mu_{\ell}.
$$
Moreover, if $\ell = o(n^{2/3})$, then
$\E{ X_\ell }
\sim \hat\mu_{\ell}$.
For any fixed $\ell$, the number $X_\ell$ of $2\ell$-cycles in $\calD$, satisfies
\[
X_\ell\overset{d}{\to} 
\Po\lp1/\ell\rp.
\]
\end{tm}
\end{tcolorbox}
\begin{proof}
Let 
\[
C = (r_{i_1}, c_{j_1}, \ldots,
r_{i_k}, c_{j_k}, \ldots,
r_{i_\ell}, c_{j_\ell}, r_{i_1})
\]
be a $2\ell$-cycle in $\calC_{2\ell}$. 
We have $\{ C \subset \calD \}$
if
the following $2\ell$ events hold
\[
\bC{
b_{i_k,j_k}
=\textstyle\displaystyle\max_{m\in[n]}b_{i_k,m}
}_{k=1}^{\ell}\quad\text{and}\quad
\bC{
a_{i_{k+1},j_k}
=\textstyle\displaystyle\max_{m\in[n]}a_{m,j_k}
}_{k=1}^{\ell}.
\]
These events involve $\ell$ distinct rows of $\bm B$ and $\ell$ distinct columns of $\bm A$. These events are mutually independent, and each holds with probability $1/n$, so
$\P{ C \subset\calD }=n^{-2\ell}$.
Thus,
\begin{align}\label{alternative writing of mu}
\E{ X_\ell }
= \Abs{ \calC_{2\ell} }
\cdot \P{ C \subset\calD }
=\frac{(n)_\ell^2}{\ell}n^{-2\ell}
=\mu_{\ell},
\end{align}
where the last step follows from the definition in \eqref{muDef}.

Next we obtain the asymptotic distribution of the short cycle count.
Note that,
\begin{eqnarray*}
\E{X_\ell}
=\mu_{\ell}
\sim    1/\ell.
\end{eqnarray*}
To prove the second part, we use the method of moments (see, e.g., \cite[Theorem 6.2]{janson2011random}).
Let $C_1,\ldots,C_k\in \calC_{2\ell}$
be $k$ distinct cycles.
By Lemma \ref{vertex-disjoint lemma},
to calculate the $k$-th factorial moments,
it suffices to consider the case where all of these $k$ cycles are pairwise disjoint, since otherwise, we have
$
    \P{ \prod_{i\in[k]} Y_{C_i} 
    =1 }=0.
$
Hence, the $k$-th factorial moment of $X_\ell$ is
\begin{align*}
\E{ (X_\ell)_k }
=\sum_{\text{disjoint } C_1,\ldots,C_k\in \calC_{2\ell} }
\E{ \prod_{i\in[k]} Y_{C_i} } 
&=\frac{1}{\ell^kn^{2k\ell}}\prod_{i=0}^{k\ell-1}(n-i)^2 \\
&=\frac{1}{\ell^k}\prod_{i=0}^{k\ell-1}\left(1-\frac{i}{n}\right)^2
\sim
\lp1/\ell\rp^k.
\end{align*}
By the method of factorial moments, we have $X_\ell\overset{d}{\to} \Po\lp1/\ell\rp$.
\end{proof}

Theorem~\ref{2lcyclesInd} implies that $\P{X_\ell\ge 1} \to 1-e^{-1/\ell}$ for fixed $\ell$. For $\ell=1$, this is just the limiting probability that a pure Nash equilibrium appears in $\calD$. 
Similar results on the probability of existence of cycles were also studied by Pangallo, Heinrich, and Farmer \cite{pangallo2019best}.


\subsection{Longest cycle length}

Next we show that
the length $2\ell_{\max}$ of the longest cycle is of order $\sqrt n$, in probability.
\begin{tcolorbox}
\begin{theorem}
\label{thm:longest}
Let $\calD$ be the best-response digraph
of an $n \times n$ random independent bimatrix game.
Then 
\[
\ell_{\max}=\Theta_\mathbb{P}\lp\sqrt{n}\rp\,.
\]
\end{theorem}
\end{tcolorbox}

\begin{proof}
Let $0<\ep<1$ be fixed.
Define 
\begin{align}
\label{cCdef}
c= e^{-2e/\ep}
\quad\mbox{and}\quad 
C= \sqrt{\log ( 1+\ep^{-1} )}.
\end{align}
Note that $0<c<C$.
It suffices to show that,
as $n\to\infty$,
\begin{align}
\label{max}
\P{ c\sqrt n \le \ell_{\max} < C\sqrt n } \ge 1-2\ep.
\end{align}

By the first moment method,
the probability that a cycle of length greater than $C\sqrt{n}$ exists is
\begin{align*}
    \P{ \sum_{\ell\ge C\sqrt{n}} X_\ell \ge 1}
\le \sum_{\ell\ge \lceil C\sqrt{n}\, \rceil} \E{X_\ell}
&\le \sum_{i\ge0}\sum_{\ell= 2^i\lceil C\sqrt{n}\, \rceil}^{2^{i+1}\lceil C\sqrt{n}\, \rceil-1} \frac1\ell e^{-\frac{(\ell-1)^2}{n}}\\
&\le \sum_{i\ge0}e^{-4^iC^2} 
\le \sum_{i\ge1}e^{-i C^2}
= \frac{e^{-C^2}}{1 - e^{-C^2}}.
\end{align*}
In view of the definition of $C$, we have
\begin{align}
\label{max-1}
\P{ \ell_{\max} \ge C\sqrt n} 
= \P{ \sum_{\ell\ge C\sqrt{n}} X_\ell \ge 1}
\le \ep.
\end{align}

Next, we use the second moment method to show
\begin{align}
\label{max-2}
\P{ \ell_{\max} < c\sqrt n } 
\le
\P{ \sum_{C\in\calT} Y_C
= 0}
\le \ep,
\end{align}
\mbox{where}
\[
\calT
=
\bC{
C\in \calC_{2\ell}:
\ceil{ c\sqrt n }
\le\ell\le
\ceil{ \sqrt n }
}.
\]
To use Lemma \ref{2rd-moment},
noting that 
any pair of cycles must be vertex-disjoint to appear in $\calD$,
\[
\Delta
=
\sum_{
C_i \cap C_j=\emptyset }
\E{ Y_{C_i}Y_{C_j} }
\le
\sum_{C_i, C_j\in\calT}
\E{ Y_{C_i} } \E{ Y_{C_j} }
= \mu^2.
\]
Choosing $\alpha=0$ and $\beta=0$
in \eqref{probX=0} yields
\[
\P{ \sum_{C\in\calT} Y_C
= 0}
\le 1/\mu.
\]
Since $c<1$ and $\ell=o(n^{2/3})$, 

\begin{align*}
\mu
=\sum_{C\in\calT} \E{ Y_C}
\sim 
\sum_{\ell= \ceil{ c\sqrt n }}^{
\ceil{ \sqrt n }}
\frac1\ell e^{-\ell^2/n}
\ge
(1+o(1))
e^{-1}
\sum_{\ell= \ceil{ c\sqrt n }}^{
\ceil{ \sqrt n }}
\frac1\ell
&\sim e^{-1} \log(1/c),
\end{align*}
and consequently,
\eqref{max-2} follows by the definition of $c$ in \eqref{cCdef}.
By the union bound, \eqref{max} follows from
\eqref{max-1} and \eqref{max-2}.
The proof is complete.
\end{proof}

\subsection{Total cycle count}
\begin{tcolorbox}
\begin{theorem}
\label{number of cycles independent case}
Let $\calD$ be the  
best-response digraph
of an
$n \times n$
random independent bimatrix game,
and $T_n$ denote the total number of directed cycles in $\calD$.
Then
\[
    \E{ T_n }
    =\lp\frac{1}{2}+o(1)\rp \log n\,.
\]
\end{theorem}
    
\end{tcolorbox}
\begin{proof}
Define the sets of integers,
\begin{align*}
I_1 \defeq 
\bS{ 1, \sqrt{ \dfrac{n}{\log n}} }
\cap \mathbb Z,
\quad
I_2 \defeq \Big( 
\sqrt{ \dfrac{n}{\log n}}, 
\sqrt{ n\log n} \Big]
\cap \mathbb Z,
\quad
I_3 \defeq \Big( \sqrt{n\log n}, n\Big]
\cap \mathbb Z.
\end{align*}
We show that the summation over $I_1$ is the dominating contribution; that is,
\begin{align*}
\E{ T_n }
= \sum_{\ell \in (I_1 \cup I_2 \cup I_3) } \E{ X_\ell }
=
\sum_{\ell\in I_1} 
\mu_{\ell}
+ \sum_{\ell\in I_2}
\mu_{\ell}
+ \sum_{\ell\in I_3}\mu_{\ell}
\sim 
\sum_{\ell \in I_1} 
\mu_{\ell} \sim \dfrac12\log n\,.
\end{align*}


\noindent
\textbf{Case 1: $\ell\in I_1$}. The
Weierstrass product inequality
gives 
$\prod_{i=1}^{\ell-1}
    \bR{1-\frac{i}{n}}^2
\ge 1-2\sum_{i=1}^{\ell-1}\frac{i}{n}$, 
so we have, in view of \eqref{muDef}, that
\begin{align*}
      \mu_\ell
\ge\frac{1}{\ell}\lp1-2\sum_{i=1}^{\ell-1}\frac{i}{n}\rp
\ge\frac{1}{\ell}-\frac{\ell}{n}.
\end{align*}
Moreover, 
\begin{align*}
\sum_{\ell\in I_1}\mu_{\ell}
\ge \log\lp\sqrt{\frac{n}{\log n}}\rp
-\frac{(\sqrt{n/\log n}+1)^2}{2n} 
\sim\frac{1}{2} \log n.
\end{align*}

For a matching upper bound, using \eqref{mubound} and noting that $\ell^2=o(n)$, we have
\begin{align*}
    \sum_{\ell\in I_1}\mu_{\ell}&\leq \sum_{\ell\in I_1}\frac1{\ell} \Exp{-\frac{(\ell-1)^2}{n}}
\le
\sum_{\ell\in I_1}
\frac{1}{\ell}
\sim\frac12\log n.
\end{align*}

\noindent
\textbf{Case 2: $\ell\in I_2$}.
By using the upper bound in \eqref{mubound}, we obtain
\begin{align*}
\sum_{\ell\in I_2}\mu_{\ell}
\leq 
\sum_{\ell\in I_2}\frac1
\ell \Exp{-\frac{(\ell-1)^2}{n}}
\le
\sum_{\ell\in I_2} 
\frac{1}{\ell}
\sim 
\log\log n.
\end{align*}

\noindent
\textbf{Case 3: $\ell\in I_3$}.
Similar to the above case, noting that 
$\ell \ge \sqrt{n\log n}$, we get
\begin{align*}
\sum_{\ell\in I_3} \mu_{\ell}
\le
\sum_{\ell\in I_3} \frac{1}{\ell}
\Exp{ -\frac{(\ell-1)^2}{n} }
\le \frac{(1+o(1))n}{ \sqrt{n\log n} }
e^{ - \log n }
\sim  \frac{1}{ \sqrt{n\log n} } = o(1).
\end{align*}
Combining the above estimates finishes the proof.
\end{proof}

\section{Random symmetric games}
In this section, we study random symmetric games, as defined in Section \ref{section:types of games}. 
We obtain the expected number of cycles of a fixed length, the expected total number of cycles, and the typical length of the longest cycle.

\subsection{Counts of two types of symmetric cycles}
We start with a simple observation that in symmetric games, whenever BR($r_i$)=$c_j$, then BR($c_i$)=$r_j$. 
This gives us pairs of cycles that are symmetric with respect to the diagonal of the payoff matrix. In particular, the 
off-diagonal pure Nash equilibria come in pairs; that is, if $(r_i,c_j)$ is an equilibrium, then $(r_j,c_i)$ is also an equilibrium.

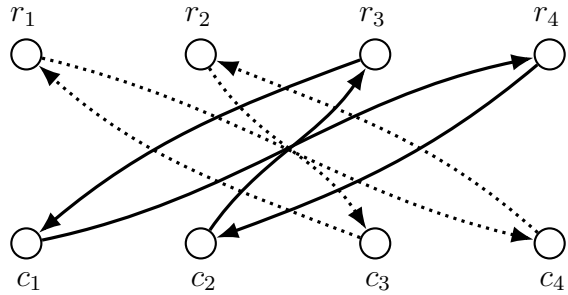
\begin{figure}[htb]
\centering
\begin{tikzpicture}[x=2.5cm,y=3cm,>=Latex,rotate=90]
  \tikzset{
    left/.style={ellipse,draw,thick,minimum width=0.1cm,minimum height=0.1cm},
    right/.style={ellipse,draw,thick,minimum width=0.1cm,minimum height=0.1cm},
    blueedge/.style={->,very thick,dotted},
    rededge/.style={->,very thick}
  }
  \foreach \i in {1,...,4}{
    \node[left]  (a\i) at (0,-\i/1.3) {};
    \node[right] (b\i) at (1,-\i/1.3) {};
    \node[anchor=east] at (-0.2,-\i/1.3-0.12) {$c_\i$};
    \node[anchor=west] at (1.2,-\i/1.3+0.12) {$r_\i$};
  }
  \draw[blueedge] (b2) to[out=-139,in=41]  (a3);
  \draw[blueedge] (a3) to[out=75,in=-126]  (b1);
  \draw[blueedge] (b1) to[out=-99,in=81]   (a4);
  \draw[blueedge] (a4) to[out=54,in=-106]  (b2);
  \draw[rededge] (a1) to[out=-74,in=106]  (b4);
  \draw[rededge] (b4) to[out=136,in=-64]  (a2);
  \draw[rededge] (a2) to[out=-25,in=155]  (b3);
  \draw[rededge] (b3) to[out=116,in=-44]  (a1);
\end{tikzpicture}
\caption{If the cycle $C$ with solid arcs appears in $\calD$, then its mirror image $C^*$, the dotted cycle is forced in $\calD$.}
\label{fig:auto-creation of cycles}
\end{figure}

A cycle of length larger than $2$ containing arc $(r_i, c_i)$
cannot appear in $\calD$
because the symmetry forces a shorter cycle, 
that is a pure Nash equilibrium.
It is useful to define the mirror image of $C$ as $C^*$, that is,
given a cycle
\[
C = (r_{i_1}, c_{j_1}, \ldots,
r_{i_k}, c_{j_k}, \ldots,
r_{i_\ell}, c_{j_\ell}, r_{i_1})
\in \calC_{2\ell},
\]
we define
\[
C^* 
=
(c_{i_1}, r_{j_1}, \ldots,
c_{i_k}, r_{j_k}, \ldots,
c_{i_\ell}, r_{j_\ell}, c_{i_1}).
\]
Note that in symmetric games,
$\{ C \subset \calD \}
= \{ C^* \subset \calD \}$.
If $C=C^*$, we call $C$ a self-mirror cycle.
\begin{tcolorbox}
\begin{lemma}
\label{sym-cycle-prob}
Let $\calD$ be the  
best-response digraph
of an
$n \times n$
random symmetric bimatrix game.
Then for any $C \in \calC_{2\ell}$,
\[
\P{ C \subset \calD }
=
\begin{cases}
n^{-2\ell}, & \text{if } \calI(C) \cap \calJ(C) = \emptyset, \\
n^{-\ell},  & \text{if } C = C^{*}, \\
0,          & \text{otherwise}.
\end{cases}
\]
\end{lemma}
\end{tcolorbox}

\begin{proof}
\textbf{Case 1: $\calI(C) \cap \calJ(C) = \emptyset$.}
Similarly to the independent case,
a $2\ell$-cycle 
$
C = (r_{i_1}, c_{j_1}, \ldots,
r_{i_k}, c_{j_k}, \ldots,
r_{i_\ell}, c_{j_\ell}, r_{i_1})
$
appears in $\calD$
if
the following $2\ell$ events hold
\begin{align}
\label{symEvents}
\bC{
b_{i_k,j_k}
=\textstyle\displaystyle\max_{m\in[n]}b_{i_k,m}
}_{i_k \in \calI}\quad\text{and}\quad
\bC{
a_{i_{k+1},j_k}
=\textstyle\displaystyle\max_{m\in[n]}a_{m,j_k}
}_{j_k \in \calJ}.
\end{align}
Since $\bm B= \bm A^{\mathsf T}$,
the above events are equivalent to
\[
\bC{
a_{j_k, i_k}
=\textstyle\displaystyle
\max_{m\in[n]}a_{m, i_k}
}_{i_k \in \calI}
\quad\text{and}\quad
\bC{
a_{i_{k+1},j_k}
=\textstyle\displaystyle\max_{m\in[n]}a_{m,j_k}
}_{j_k \in \calJ}.
\]
These events involve columns in $\calI$ and $\calJ$
of $\bm A$ respectively, and because 
$\calI \cap \calJ =\emptyset$
and $\bm A$
has i.i.d.\ continuous entries,
the events are mutually independent, and each holds with probability $1/n$, so
$\P{ C \subset\calD }=n^{-2\ell}$.

\textbf{Case 2: $C^*=C$.}
We show that
the first $\ell$ events in \eqref{symEvents}
imply the other $\ell$ events.
To see this, we claim that for each
arc $(r_{i}, c_{j}) \in C$ such that
$b_{i,j}
=\max_{m\in[n]}b_{i,m}$,
there exists a mirror pair 
$(c_{i}, r_{j}) \in C^*$
such that
$a_{j, i}
= \max_{m\in[n]}a_{m, i}$,
and because $\bm B= \bm A^{\mathsf T}$,
these are the same events,
and therefore,
\[
\{ C \subset \calD \}
=\{ C^* \subset \calD \}
= 
\bigcap_{j_k \in \calJ}\bC{
a_{i_{k+1},j_k}
=\max_{m\in[n]}a_{m,j_k}
}.
\]
These $\ell$ events involve distinct columns of $\bm A$, so they are mutually independent
and each holds with probability $1/n$, so $\P{C \subset \calD}=n^{-\ell}$.
This completes the second case.

\textbf{Case 3:} $\calI(C) \cap \calJ(C) \ne \emptyset$ and $C^*\ne C$.
Let $u \in \calI(C) \cap \calJ(C)$.
Then we have 
$u \in \calI(C^*)$.
So vertex $r_u$ is on both $C$ and $C^*$.
By Lemma \ref{vertex-disjoint lemma},
we get $\P{C\subset \calD}=0$.
\end{proof}

\begin{lemma}
\label{mirror-count}
For $\ell\ge1$,
\[
\Abs{ \{ C \in \calC_{2\ell}: C=C^*\}}
=
\frac{(n)_\ell}{\ell} \, \I{ \ell \text{ is odd} }.
\]
\end{lemma}

\begin{proof}
Fix any arc $e$ of a self-mirror cycle $C \in \calC_{2\ell}$, and let $x$ be the minimum
number of steps along $C$ required to reach its mirror arc $e^*$ in cycle $C$.
By symmetry, mirroring this path gives a path of the same length $x$
from $e^*$ back to $e$. These two paths partition the
$2\ell$-cycle $C$,
because $x$ is the minimum
number of steps to reach $e^*$.
Hence, $x=\ell$. Thus, every arc and its mirror are separated by
exactly $\ell$ steps.
Since every step along $C$ switches sides
and the initial vertices of $e$ and $e^*$ are on opposite sides
of the vertex bipartition, $\ell$ must be odd.
See Figure \ref{fig:self-mirror} for an example.
\begin{figure}[htbp]
\centering
\begin{tikzpicture}[
scale=0.85,
  rv/.style={rectangle,draw,fill=white,minimum size=7mm,inner sep=1pt,font=\small},
  cv/.style={circle,draw,fill=white,minimum size=7.5mm,inner sep=0.5pt,font=\small},
  rvb/.style={rv,line width=1.4pt},
  cvb/.style={cv,line width=1.4pt},
  arr/.style={-{Stealth[length=2.2mm]},semithick,black!45},
  hot/.style={-{Stealth[length=2.6mm]},line width=1.5pt},
  mir/.style={densely dashed,black!30},
  hotmir/.style={densely dashed,line width=1.2pt,black!75}
]
\def\R{2.6}
\foreach \t/\sty/\lab in {0/rvb/{r_1},1/cvb/{c_4},2/rv/{r_2},3/cv/{c_5},4/rv/{r_3},
                          5/cvb/{c_1},6/rvb/{r_4},7/cv/{c_2},8/rv/{r_5},9/cv/{c_3}}
  \node[\sty] (v\t) at ({90-36*\t}:\R) {$\lab$};
\foreach \t [evaluate=\t as \n using {int(mod(\t+1,10))}] in {1,2,3,4,6,7,8,9}
  \draw[arr] (v\t) -- (v\n);
\foreach \a/\b in {2/7,4/9,8/3}
  \draw[mir] (v\a) -- (v\b);
\draw[hotmir] (v0) -- (v5);
\draw[hotmir] (v1) -- (v6);
\draw[hot] (v0) -- (v1);
\draw[hot] (v5) -- (v6);
\end{tikzpicture}
\caption{A self-mirror directed cycle $C=C^*
\in\calC_{2\ell}$ with $\ell=5$.
After $5$ steps along $C$, arc $e=(r_1, c_4)$ moves to its mirror arc $e^*=(c_1, r_4)$.}
\label{fig:self-mirror}
\end{figure}
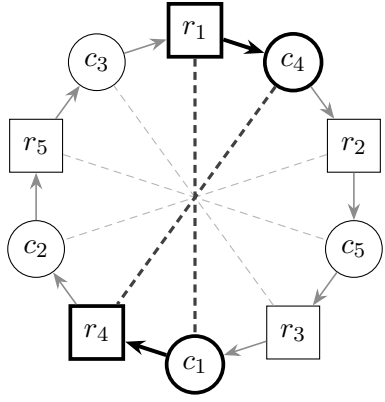

Note that
a $2\ell$-cycle corresponds to a cyclic permutation and
there are $(n)_\ell$ ordered choices,
and each cyclic ordering is counted $\ell$ times according to its starting point.
This completes the proof.
\end{proof}

By Lemma \ref{sym-cycle-prob}, 
to count cycles in $\calD$,
it suffices to consider 
self-mirror $2\ell$-cycles
$\{ C \in \calC_{2\ell}: C=C^* \}$
and cycles $C \in \calC_{2\ell}^{0}$
with disjoint index sets 
$\calI(C) \cap \calJ(C) = \emptyset$
of strategies,
where
\begin{align}
    \label{C0def}
\calC_{2\ell}^{0}
:= \{ C \in \calC_{2\ell} : 
Z(C) = 0 
\}.
\end{align}
Then we have
\begin{align}
\label{sym-Xell}
X_\ell = U_\ell \cdot \I{\ell\text{ is odd}} + V_\ell,
\quad\mbox{where}\quad
U_\ell = \sum_{C \in \calC_{2\ell}: C = C^* } \I{C \subset \calD},
\quad
V_\ell = \sum_{C \in \calC_{2\ell}^0 } \I{C \subset \calD}.
\end{align}

\begin{tcolorbox}
    \begin{theorem}
\label{symExpected}
       Let $\calD$ be the  
best-response digraph
of an
$n \times n$
random symmetric bimatrix game.
Let $X_\ell$ count the number of $2\ell$-cycles in $\calD$.
Then for $1\le\ell\le n$,
$$
\E{X_\ell}
= 
\frac{1}{\ell}\prod_{i=0}^{\ell-1}\lp 1-\frac{i}{n}\rp
\I{\ell \text{ is odd}}
+ 
\frac{1}{\ell}\prod_{i=0}^{2\ell-1}\lp1-\frac{i}{n}\rp.
$$
Moreover, if $\ell = o(n^{2/3})$, then
\begin{align}
\label{sym-exp-asympt}
\E{ X_\ell }
\sim
\frac{1}{\ell}\exp\lp-\frac{\ell^2}{2n}\rp
\I{\ell \text{ is odd}}.
+
\frac{1}{\ell}\exp\lp-\frac{2\ell^2}{n}\rp.
\end{align}
    \end{theorem}
\end{tcolorbox}

\begin{proof}
First note that
\[
\abs{\calC_{2\ell}^{0}}
= \binom{n}{2\ell}\binom{2\ell}{\ell}\frac{\ell!^2}{\ell}
= \frac{ (n)_{2\ell} }{\ell},
\]
as to count the cycles on disjoint sets of strategies, 
there are $\binom{n}{2\ell}$ choices for the 2 disjoint sets of strategies, $\binom{2\ell}{\ell}$ choices for the one cycle, and we order both sets and divide by $\ell$ for the overcounting.
Then using Lemma \ref{sym-cycle-prob} for $C \in \calC_{2\ell}^{0}$, we have
\[
\E{ V_\ell }
= \abs{\calC_{2\ell}^{0}}
n^{-2\ell}
= \frac{ (n)_{2\ell} }{\ell}
n^{-2\ell}.
\]
By Lemma \ref{mirror-count} and Lemma \ref{sym-cycle-prob},
for self-mirror
cycles,
\[
\E{ U_\ell } 
= \Abs{ \{ C \in \calC_{2\ell}: C=C^*\}}
n^{-\ell} \I{\ell\text{ is odd}}
= \frac{(n)_\ell}{\ell}
n^{-\ell} \I{\ell\text{ is odd}}.
\]
In view of \eqref{sym-Xell}, combining the above yields the first part.
The derivation of the asymptotics for $\ell=o(n^{2/3})$ is similar to that in 
Lemma \ref{lem:muEquiv}.
\end{proof}

Next we determine the asymptotic distribution of the number of $2\ell$-cycles. The limiting distribution depends on whether $\ell$ is even or odd.

\begin{tcolorbox}
\begin{theorem}
Let $\calD$ be the  
best-response digraph
of an
$n \times n$
random symmetric bimatrix game.
For $\ell\ge2$, 
let $U_\ell$ and $V_\ell$ be defined by \eqref{sym-Xell}.
Then
\[
(U_\ell, V_\ell/2)
\stackrel{d}{\to} 
\bR{W_1 \cdot \I{\ell\text{ is odd}}, W_2}
\quad\mbox{where}\quad
W_1 \sim \Po(1/\ell)
\quad\mbox{and}\quad
W_2 \sim \Po(1/(2\ell))
\]
are independent random variables.
Let 
$X_\ell$ count the number of $2\ell$-cycles in $\calD$. Then
\[
X_{\ell} 
\stackrel{d}{\to} 
W_1  \cdot \I{\ell\text{ is odd}} + 2 W_2.
\]

\end{theorem}
\end{tcolorbox}

Note that if $\ell$ is odd,
$U_\ell$ and $V_\ell$ count
the number of $2\ell$-cycles with the same and disjoint sets of strategy indices, respectively.
Since every cycle $C$ counted by $V_\ell$ appears together with its distinct mirror $C^*$, 
then $V_\ell/2$ is the mirror (cycle) pair count.

\begin{proof}
First, assume $\ell\ge2$ is even. 
By Lemma \ref{mirror-count},
we have the self-mirror cycle count $U_\ell=0$.
Hence it suffices to consider the cycles in $\calC_{2\ell}^{0}$,
recalling its definition in \eqref{C0def}.
Let $\calP_{2\ell}^{0}$ be the set of mirror pairs $P=\{C,C^*\}$ for $C\in\calC_{2\ell}^{0}$, and define the mirror pair count
$Y_P$ by
\[
\calP_{2\ell}^{0}
= \bC{ \{C,C^*\}:  C \in \calC_{2\ell}^{0} }
\quad\mbox{and}\quad
Y_P := Y_C=Y_{C^*}.
\]
Then
\[
V_\ell/2 = \sum_{ P\in\calP_{2\ell}^{0} } Y_P,
\quad\mbox{and}\quad
\abs{\calP_{2\ell}^{0}}=\frac12\abs{\calC_{2\ell}^{0}}.
\]
We consider all of $k$ mirror pairs to use pairwise disjoint sets of strategies, since otherwise the
probability that $\prod_{j\in[k]} Y_{P_j} = 1$ is zero. 
Thus, the random variables $\{ Y_{P_j}: j \in [k] \}$ are independent and 
\begin{align*}
\E{ (V_\ell/2)_k}
=\sum_
{ (P_1,\ldots,P_k )\text{ disjoint}}
\E{ \prod_{j\in[k]}Y_{P_j}}
\,&=\prod_{i=0}^{k-1}\lp\frac12 \binom{n-2i\ell}{2\ell}\binom{2\ell}{\ell}\frac{\ell!^2}{\ell}n^{-2\ell}\rp\\
    &=\frac{1}{(2\ell)^k}\prod_{i=0}^{2k\ell-1}\left(1-\frac{i}{n}\right)\to \lp1/2\ell\rp^k
\end{align*}
as $n\to\infty$. 
Hence 
$V_\ell/2
\stackrel{d}{\to} W_2$ and $(U_\ell,V_\ell/2)=(0,V_\ell/2)
\stackrel{d}{\to} (0,W_2)$.

Next, fix an odd $\ell\ge3$.
Similarly, by the method of moments, it suffices to show
\begin{align*}
\E{\binom{U_\ell }{r}\binom{V_\ell/2}{s}}
&=\frac{1}{n^{\ell(r+2s)}}
\prod_{i=0}^{r-1}
\binom{n-i\ell}{\ell}
\lp\frac{\ell!}{\ell}\rp^r\frac{1}{r!}
\cdot
\prod_{i=0}^{2s-1}
\binom{n-(r+i)\ell}{\ell}
\lp\frac{\ell!^2}{2\ell}\rp^s\frac{1}{s!}\\
&=\prod_{i=0}^{\ell(r+2s)-1}\lp1-\frac{i}{n}\rp\frac{1}{s!r!}\frac{1}{\ell^{r}(2\ell)^s}
    \to \frac{1}{\ell^r r!}\frac{1}{(2\ell)^s s!},
\end{align*}
as $n\to\infty$.
To see this, the choice of $r$ self-mirror cycles 
and $s$ mirror pairs with pairwise disjoint index sets occurs with probability $n^{-\ell r-2\ell s}$. 
To count such choices, take an ordered list of $\ell(r+2s)$ distinct indices. Each self-mirror cycle is counted $\ell$ times, each mirror pair $2\ell$ times, and ordering the objects contributes factors $r!$ and $s!$.

Since $X_\ell = U_\ell + 2(V_\ell/2)$ for all $\ell\ge2$ by Lemma \ref{sym-cycle-prob}, we get $X_\ell \stackrel{d}{\to} W_1\cdot\I{\ell\text{ is odd}} + 2W_2$. 
\end{proof}

\subsection{Longest cycle length}
The length of the longest cycle in a random symmetric game is similar to that in the independent case.
\begin{tcolorbox}
\begin{theorem}
Let $\calD$ be the  
best-response digraph
of an
$n \times n$
random symmetric bimatrix game.
The length of the longest cycle is 
\[
2\ell_{\max}=\Theta_\mathbb{P}\lp\sqrt{n}\rp\,.
\]
\end{theorem}
\end{tcolorbox}
\begin{proof}
Let $0<\ep<1$ be fixed,
and define
\[
c=e^{-4e^2/\ep}
\quad\mbox{and}\quad
C= \sqrt{2\log ( 1+2/\ep )}.
\]
The first moment calculation is similar to the proof of Theorem \ref{thm:longest}.
By Theorem \ref{symExpected}, 
the expected cycle count is
\[
\E{X_\ell}
\le 
\frac{2}{\ell}\prod_{i=0}^{\ell-1}\lp 1-\frac{i}{n}\rp
\le 
\frac{2}{\ell}\Exp{-\frac{(\ell-1)^2}{2n}},
\]
and bounding the summation of $\E{X_\ell}$ over $\ell \ge \ceil{C\sqrt n}$ yields 
$\P{ \ell_{\max} \ge C\sqrt{n} }\le\ep$.

For the lower bound, 
similarly to the proof of Theorem \ref{thm:longest}, define
\[
\calT
=
\bC{
C\in \calC_{2\ell}^0:
\ceil{ c\sqrt n }
\le\ell\le
\ceil{ \sqrt n }
}
\]
to be a set of cycles with disjoint index sets.
Then we have
\[
\P{ \ell_{\max} < c\sqrt{n}}
\le
\P{ \sum_{C\in\calT}Y_C=0} \le 
2/\mu
\quad\mbox{ with }\quad
\mu=\sum_{C\in\calT} \E{Y_C},
\]
by choosing $\alpha=1$ and $\beta=0$
in Lemma \ref{2rd-moment}.
To justify the choice of $\alpha$ and $\beta$, 
note that
\[
\Delta
=
\sum_{
\substack {C_i\ne C_j\in\calT \\
C_i \cap C_j=\emptyset}
}
\E{ Y_{C_i}Y_{C_j} }
\le
\sum_{C\in\calT}\E{Y_C}
+\sum_{\substack{C\ne C'\in\calT\\ 
C'\ne C^*, C'\cap C=\emptyset }}\E{Y_C}\E{Y_{C'}}
\le \mu+\mu^2,
\]
where we split the summation by
separating the case of mirror pairs,
and noting that if $C=C^*$,
\[
\E{Y_C Y_{C^*} }=\E{ Y_C }.
\]
Noting that if $C \ne C' \ne C^*$ 
are all distinct
and $C \cap C'=\emptyset$, then 
$\{ C \in \calD \}$
and $\{ C' \in \calD \}$ are independent events,
as they involve disjoint sets of entries of the payoff matrix $\bm A$.
In view of the asymptotics of cycle count in
\eqref{sym-exp-asympt},
keeping only the even $\ell$ 
as a lower bound on the cycle count,
we have
\[
\mu
\ge 
(1+o(1))\sum_{\ell=\ceil{ c\sqrt n }}
^{\ceil{ \sqrt n }} e^{-2\ell^2/n} \frac{1}{\ell}
\ge (1+o(1)) e^{-2}\log(1/c)
\ge (1+o(1)) 4/\ep,
\]
and therefore,
$
\P{ \ell_{\max} < c\sqrt{n}}
\le
\ep$. The proof is complete.
\end{proof}

\subsection{Total cycle count}

We close this section with an estimate of the total number of cycles in a random symmetric bimatrix game.
The proof is similar to that of Theorem \ref{number of cycles independent case}.

\begin{tcolorbox}
    \begin{theorem}
\label{number of cycles-symmetric}
Let $\calD$ be the  
best-response digraph
of an
$n \times n$
random symmetric bimatrix game,
and $T_n$ denote the total number of directed cycles in $\calD$.
Then
\[
    \E{ T_n } =\lp\frac{3}{4}+o(1)\rp  \log n\,.
\]
\end{theorem}
\end{tcolorbox}

\begin{proof}
We split the summation according to the parity of $\ell$,
\begin{align*}
\E{ T_n }
= \sum_{\ell\text{ odd}} 
\E{X_\ell}
+ \sum_{\ell\text{ even}} \E{X_\ell}
=\sum_{\text{odd } \ell\ge1} 
\frac{1}{\ell}\prod_{i=0}^{\ell-1}\lp1-\frac{i}{n}\rp
+ \sum_{\ell\ge1} 
\frac{1}{\ell}\prod_{i=0}^{2\ell-1}\lp1-\frac{i}{n}\rp\,.
\end{align*}
The second sum, following a similar analysis to that in the proof of Theorem \ref{number of cycles independent case}, is asymptotic to $\log n/2$. 
The first sum can be estimated as follows.
For $\ell>\sqrt{n\log n}$, we have
\begin{align*}
&\sum_{\text{odd } \ell=
\floor{\sqrt{n\log n}}}^n 
\dfrac{1}{\ell}\prod_{i=0}^{\ell-1}\lp1-\frac{i}{n}\rp \\
&\qquad\le
\sum_{\ell=\floor{\sqrt{n\log n}}}^n \dfrac{1}{\ell}
\Exp{ -\dfrac{\ell(\ell-1)}{2n} }
\le (1+o(1)) \frac{n e^{ - \log n /2}}{ \sqrt{n\log n} }
\sim \log^{-1/2} n.
\end{align*}
For the rest of the values of $\ell$, we have
\begin{align*}
\sum_{\text{odd } \ell} \dfrac{1}{\ell}\prod_{i=0}^{\ell-1}\lp1-\frac{i}{n}\rp
&\sim
\sum_{\text{odd } \ell} \dfrac{1}{\ell}
\Exp{ -\dfrac{\ell^2}{2n} }\\
&=\sum_{k\ge1} \frac{1}{2k-1}
\Exp{ -\dfrac{(2k-1)^2}{2n} }
\sim \frac12 \sum_{k\ge1} \frac{1}{k}
\Exp{ - \frac{2k^2}{n} }
\sim \frac14\log n\,.
\end{align*}
Combining the above estimates finishes the proof.
\end{proof}

\section{Random antisymmetric games}
In this section, we study antisymmetric games, where 
the payoff matrix $\bm A$ of the first player is an $n \times n$ random matrix, whose entries are independent, drawn from the same nonatomic distribution,
and
the payoff matrix of the second player is $\bm B=-\bm A^{\mathsf{T}}$.

We estimate the expected number of cycles of a fixed length, the expected total number of cycles, and the typical length of the longest cycle.

\subsection{Pure Nash equilibria}
We start with $2$-cycles that correspond to PNE.
\begin{tcolorbox}
    \begin{theorem}
    \label{anti-cycle}
For an
$n \times n$
random antisymmetric bimatrix game,
we have
\[
\E{N_\PNE} = 1-1/n
\quad\mbox{and}\quad
N_\PNE \overset{d}{\to} \Po(1)\,.
\]
    \end{theorem}
\end{tcolorbox}

\begin{proof}
Fix a strategy $r_i$ of the row player and a strategy $c_j$ of the column player. 
By definition, the pair $(r_i,c_j)$ is a pure Nash equilibrium if 
$a_{i,j}$ is the maximum in column $j$ and $a_{j,i}$ is the minimum in column $i$. 
Thus,
$\P{(r_i,c_j)\text{ is a PNE} }
=1/n^2$.

Note that we cannot have pure Nash equilibria with $i=j$, and whenever $(r_i,c_j)$ is a pure Nash equilibrium, $(r_j,c_i)$ cannot be a pure Nash equilibrium.
Let
\[
N_{\PNE} 
= \sum_{i \ne j} \I{(r_i, c_j) \text{ is a PNE}}.
\]
Then
$$
    \E{N_\PNE}=1-\frac{1}{n}\,.
$$

To calculate the $k$-th binomial moment, we sum over $k$ pure Nash equilibria that share no rows or columns. The number of ways to choose the first point $(i_1,j_1)$ is $n(n-1)$, since it is off the diagonal. The second point $(i_2,j_2)$ has to be in a different row and column than the first one, so we have $(n-2)(n-3)$ choices for it, and so on. We also divide by $k!$ to account for the ordering of the $k$ points.

The non-disjoint terms (where some index is shared) contribute a negligible amount. In fact, if two PNE $(r_{i_1}, c_{j_1})$ and $(r_{i_2},c_{j_2})$ share an index, they impose conditions on a common column of $\bm A$, and the number of such tuples is $O(n^{2k-1})$, each with probability $O(n^{-2k})$, giving a total contribution of $O(1/n)$. 
So, for any fixed $k$,
\begin{align*}
\E{\binom{X_1}{k}}
&=\frac{1}{k!} \sum_{\text{disjoint } (i_1,j_1),\ldots,(i_k,j_k)}
\E{
\prod_{r \in [k]}  Y_{i_r,j_r} }
+ \O{1/n}\\
&=\frac{1}{k!} \prod_{i=0}^{2k-1} 
(n-i)\lp\frac{1}{n^2}\rp^k
+ \O{1/n}
=\frac{1}{k!} 
\prod_{i=1}^{2k-1} 
\bR{ 1-\frac{i}{n} } + \O{1/n}
\sim\frac{1}{k!}.
\end{align*}
By the method of moments, we have that $X_1\overset{d}{\to} \Po\lp1\rp$. 
The proof is complete.
\end{proof}

\subsection{Cycle count}
For an antisymmetric game, the best-response digraph
has a special structure due to the payoff relationship between players. When one player has the best response to a strategy, the other player's corresponding best response must be different because of the anti-symmetry and the lack of ties. That is, for $n\ge2$, whenever BR($r_i$)=$c_j$, then BR($c_i$)$\neq r_j$. 

\begin{tcolorbox}
    \begin{definition}\label{def:sympair}
In the bipartite digraph on vertex bipartition $(R, C)$, a \emph{symmetric pair} is a pair of
directed edges of the form
\[
  \bigl\{(r_{i}, c_{j}),\; (c_{i}, r_{j})\bigr\}.
\]
A directed cycle $C \in \cC_{2\ell}$ contains a symmetric pair if both
arcs of it belong to $C$.
\end{definition}
\end{tcolorbox}

Recall that we use $\calI=\calI(C)$ and $\calJ=\calJ(C)$ for
the row and column indices of cycle $C \in \calC_{2\ell}$, respectively,
and $Z = |\calI \cap \calJ|$.
\begin{tcolorbox}
\begin{lemma}\label{lem:overlap}
Let $\ell\ge2$ and $C$ be a uniformly random cycle in $\cC_{2\ell}$ and set $Z = Z(C)$. Then
\[
  \E{\left(\frac{n}{n-1}\right)^{Z}}
  = 
  1 + O\!\left(\frac{\ell^{2}}{n^{2}}\right).
\]
Let $S(C)$ be the number of symmetric
pairs contained in $C$. Then
\[
  \E{S(C)}
  = 
  \frac{\ell(\ell - 2)}{n(n - 1)}
  = 
  O\!\left(\frac{\ell^{2}}{n^{2}}\right).
\]
\end{lemma}
\end{tcolorbox}

\begin{proof}
Under the uniform measure on $\cC_{2\ell}$, 
the sets $\calI(C)$ and $\calJ(C)$ are
independent uniformly random $\ell$-subsets of $[n]$, so
$Z = |\calI(C) \cap \calJ(C)|$ follows the hypergeometric distribution with parameters $(n, \ell, \ell)$. 
Then
\[
  \E{\left(\frac{n}{n-1}\right)^{Z}}
  =
  \E{ \bR{1 - \frac{1}{n}}^{-Z}}
  \le
  \left(1 + \frac{\ell}{n(n-1)}\right)^{\!\ell}\le e^{\frac{\ell^2}{n(n-1)}}=1+ O\!\left(\frac{\ell^{2}}{n^{2}}\right),
\]
where the first inequality is by bounding the moment generating function of the hypergeometric variable $Z$ by that of a binomial variable
$\mathrm{Bin}(\ell, \ell/n)$, 
see \cite[Proof of Theorem 2.10]{janson2011random}.
Noting $Z\ge0$
implies
$\E{ (n/(n-1))^{Z}}\ge1$.




For the second part,
let
$C=(r_{i_1}, c_{j_1}, r_{i_2}, c_{j_2},\ldots, r_{i_\ell}, c_{j_\ell}) \in \calC_{2\ell}$.
A symmetric pair forms if event
\[
  \calE_{s,t}\defeq
  \{\,  i_s = j_t, \mbox{ and } i_{t+1} = j_s \,\}
\]
holds for some $s, t$. Thus, $S(C)=\sum_{s,t\in[\ell]}\I{\calE_{s,t}}$,
where the number of admissible pairs $(s,t)$ is
$\ell(\ell - 2)$,
as
there are $\ell$ choices of the arc $r_{i_s}\to c_{j_s}$,
and given $s$, there are $\ell-2$ ways to reach $t$, as
$t \in [\ell] \setminus \{ s-1, s \}$.
See Figure \ref{fig:sym-pair} for an illustration.

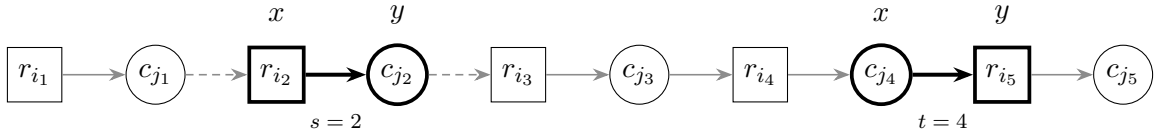
\begin{figure}[htbp]
\centering
\begin{tikzpicture}[
  x=1.6cm,y=1cm,
  rv/.style={rectangle,draw,minimum size=7.2mm,inner sep=1pt,font=\small},
  cv/.style={circle,draw,minimum size=7.8mm,inner sep=0.5pt,font=\small},
  arr/.style={-{Stealth[length=2.2mm]},semithick,black!45},
  hot/.style={-{Stealth[length=2.6mm]},line width=1.5pt},
  rvb/.style={rv,line width=1.4pt},
  cvb/.style={cv,line width=1.4pt},
  bad/.style={-{Stealth[length=2.2mm]},semithick,densely dashed,black!45}
]
\foreach \k/\lab/\sty in {0/{r_{i_1}}/rv,1/{c_{j_1}}/cv,2/{r_{i_2}}/rvb,
                          3/{c_{j_2}}/cvb,4/{r_{i_3}}/rv,5/{c_{j_3}}/cv,6/{r_{i_4}}/rv,
                          7/{c_{j_4}}/cvb,8/{r_{i_5}}/rvb,9/{c_{j_5}}/cv}
  \node[\sty] (v\k) at (\k,0) {$\lab$};
\foreach \a/\b in {0/1,4/5,5/6,6/7,8/9} \draw[arr] (v\a) -- (v\b);
\draw[bad] (v1) -- (v2);
\draw[bad] (v3) -- (v4);
\draw[hot] (v2) -- node[below=10pt,font=\scriptsize]{$s=2$} (v3);
\draw[hot] (v7) -- node[below=10pt,font=\scriptsize]{$t=4$} (v8);
\foreach \k/\v in {2/x,3/y,7/x,8/y} \node[font=\small] at (\k,0.8) {$\v$};
\end{tikzpicture}
\caption{A cycle $C\in\mathcal C_{2\ell}$ with $\ell=5$, $s=2$, and $t=4$.
Arcs $r_{i_s}\to c_{j_s}$ and $c_{j_t}\to r_{i_{t+1}}$ form a symmetric pair if both share the head index $x$ and tail index $y$.
Note that $\calE_{s,t}$ cannot happen if $t\in\{s-1, s\}$,
so there are $\ell-2$ ways to choose $t$ for each given $s$.
}
\label{fig:sym-pair}
\end{figure}

Conditioning on the row and column indices $\calI(C)$ and $\calJ(C)$,
the sequences $(i_1,\ldots,i_\ell)$ and $(j_1,\ldots,j_\ell)$ are 
independent uniformly random orderings of $\calI$ and $\calJ$, respectively.
Fix an admissible pair $(s,t)$.
There are $\ell^2(\ell-1)^2$
ways to choose
the four entries $i_s, i_{t+1}, j_t, j_s$.
The event $\calE_{s,t}$ holds if
the two arcs $r_{i_s}\to c_{j_s}$ and $c_{j_t}\to r_{i_{t+1}}$
have a common tail index $y=i_s=j_t$ and a common head index $x=j_s=i_{t+1}$.
Moreover, there are $Z(Z-1)$ ways to choose distinct $x,  y\in\calI\cap\calJ$ 
as $i_s\ne i_{t+1}$.
Hence
\[
  \P{\calE_{s,t} \mid Z}
  =
  \frac{Z(Z - 1)}{\ell^{2}(\ell - 1)^{2}}\,.
\]
Summing over the $\ell(\ell - 2)$ admissible position pairs, we have
\[
  \E{S(C) }
  = \E{ \E{S(C) \mid Z} }
  = 
  \frac{(\ell - 2)}{\ell\,(\ell - 1)^{2}} \E{ Z(Z - 1) },
\]
where $Z$ follows the hypergeometric distribution with parameters $(n, \ell, \ell)$. Since
$
  \E{Z(Z - 1)}
  = 
  (\ell)_{2}^2/ (n)_{2},
$
 the proof is complete.
\end{proof}

\begin{tcolorbox}
\begin{lemma}\label{lem:cycle_probability}
Let $n\ge2$.
For a fixed cycle $C\in\cC_{2\ell}$ with overlap size $Z=Z(C)$ and symmetric-pair count $S=S(C)$, we have
\[
\P{C\subset\calD}
=
\frac{1}{n^{2\ell}}\!\bR{\frac{n}{n-1}}^{Z} \I{ S=0}.
\]
\end{lemma}

\end{tcolorbox}

\begin{proof}
For an anti-symmetric game,
an arc $r_{i} \to c_{j}$ appears in $\calD$ if $a_{ji}$ is the minimum of
column $i$, 
and an arc $c_{i} \to r_{j}$ appears in $\calD$ if $a_{ji}$ is the maximum
of column $i$. 
Hence,
\begin{itemize}
  \item each column index in $\calI \setminus \calJ$ corresponds to one minimum
        constraint;
  \item each column index in $\calJ \setminus \calI$ corresponds to one maximum
        constraint;
  \item each column index in $\calI \cap \calJ$ corresponds to one minimum and one
        maximum constraint on different rows.
\end{itemize}

\begin{figure}[htbp]
\centering
\tikzset{
  brmax/.style={draw,line width=0.8pt,fill=white},
  brmin/.style={fill=black},
  brv/.style={circle,draw,line width=0.8pt,inner sep=0pt,minimum size=6.4mm},
  brvoff/.style={circle,draw=black!30,text=black!40,line width=0.6pt,inner sep=0pt,minimum size=6.4mm},
  rarr/.style={-{Stealth[length=2.6mm]},line width=0.8pt,shorten >=0.5pt,shorten <=0.5pt},
  carr/.style={-{Stealth[length=2.6mm]},line width=0.8pt,densely dashed,shorten >=0.5pt,shorten <=0.5pt}
}
\begin{subfigure}[b]{0.48\textwidth}
\centering
\begin{tikzpicture}[x=1cm,y=1cm,line width=0.8pt]
\path (0,0.75) (0,-4.75);
\draw[dashed,line width=0.4pt,black!55] (0.80,0) -- (0.80,-4.00);
\draw[dashed,line width=0.4pt,black!55] (0,-0.80) -- (4.00,-0.80);
\draw[dashed,line width=0.4pt,black!55] (1.60,0) -- (1.60,-4.00);
\draw[dashed,line width=0.4pt,black!55] (0,-1.60) -- (4.00,-1.60);
\draw[dashed,line width=0.4pt,black!55] (2.40,0) -- (2.40,-4.00);
\draw[dashed,line width=0.4pt,black!55] (0,-2.40) -- (4.00,-2.40);
\draw[dashed,line width=0.4pt,black!55] (3.20,0) -- (3.20,-4.00);
\draw[dashed,line width=0.4pt,black!55] (0,-3.20) -- (4.00,-3.20);
\draw (0,0) rectangle (4.00,-4.00);
\node[font=\scriptsize] at (-0.3,-0.40) {$1$}; \node[font=\scriptsize] at (0.40,0.3) {$1$};
\node[font=\scriptsize] at (-0.3,-1.20) {$2$}; \node[font=\scriptsize] at (1.20,0.3) {$2$};
\node[font=\scriptsize] at (-0.3,-2.00) {$3$}; \node[font=\scriptsize] at (2.00,0.3) {$3$};
\node[font=\scriptsize] at (-0.3,-2.80) {$4$}; \node[font=\scriptsize] at (2.80,0.3) {$4$};
\node[font=\scriptsize] at (-0.3,-3.60) {$5$}; \node[font=\scriptsize] at (3.60,0.3) {$5$};
\fill[brmin] (0.40,-2.00) circle (3.5pt);
\fill[brmin] (2.80,-0.40) circle (3.5pt);
\draw[brmax] (0.40,-0.40) circle (3.5pt);
\draw[brmax] (2.00,-2.80) circle (3.5pt);
\end{tikzpicture}
\end{subfigure}
\begin{subfigure}[b]{0.48\textwidth}
\centering
\begin{tikzpicture}[x=1cm,y=1cm]
\path (0,0.75) (0,-4.75);
\node[brv] (r1) at (0,-0.40) {$r_{1}$};
\node[brv] (c1) at (3.6,-0.40) {$c_{1}$};
\node[brv] (r2) at (0,-1.20) {$r_{2}$};
\node[brv] (c2) at (3.6,-1.20) {$c_{2}$};
\node[brv] (r3) at (0,-2.00) {$r_{3}$};
\node[brv] (c3) at (3.6,-2.00) {$c_{3}$};
\node[brv] (r4) at (0,-2.80) {$r_{4}$};
\node[brv] (c4) at (3.6,-2.80) {$c_{4}$};
\node[brv] (r5) at (0,-3.60) {$r_{5}$};
\node[brv] (c5) at (3.6,-3.60) {$c_{5}$};
\draw[rarr] (r1) -- (c3);
\draw[carr] (c3) -- (r4);
\draw[rarr] (r4) -- (c1);
\draw[carr] (c1) -- (r1);
\end{tikzpicture}
\end{subfigure}
\caption{
The payoff matrix $\bm A$ of an antisymmetric game with $n=5$. Since $\bm B=-\bm A^{\mathsf T}$, the arc $r_i\to c_j$ is present if and only if $a_{j,i}$ is the minimum of column $i$, and the arc $c_j\to r_i$ is present if and only if  $a_{i,j}$ is the maximum of column $j$. For cycle $C=(r_1,c_3,r_4,c_1)$ with $\ell=2$, $\calI=\{1,4\}$, $\calJ=\{1,3\}$, 
the shared index is $1$, so $Z(C)=1$.
Column $1$ corresponds to a maximum and a minimum condition with probability $1/(n(n-1))$, 
and columns $3$ and $4$ correspond to one condition each.
Hollow circles are maxima and solid dots are minima of the corresponding columns.}
\label{fig:anti-cycle}
\end{figure}
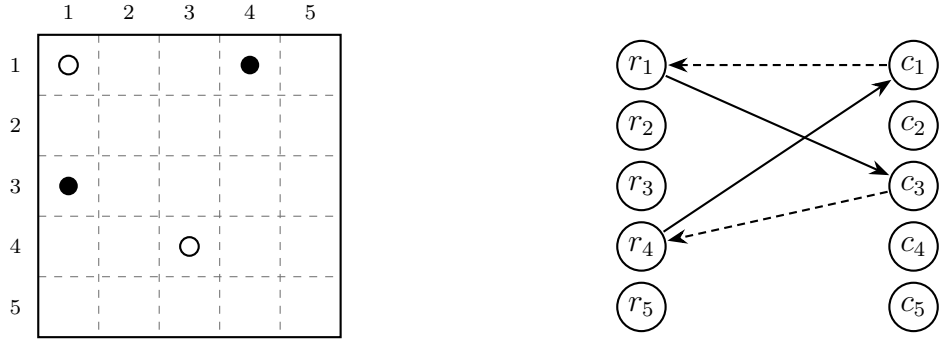

See Figure \ref{fig:anti-cycle} for an example.
There are $2(\ell - Z)$ singly-constrained columns and $Z$ doubly-constrained columns.
A singly-constrained column contributes probability $1/n$. 
For a doubly-constrained column, 
when $S=0$,
the maximum and minimum occur in
different rows, so it contributes probability $1/(n(n - 1))$. 
Therefore,
\[
  \P{C \subset \calD}
  = 
  \left(\frac{1}{n}\right)^{\!2(\ell - Z)}
  \!\left(\frac{1}{n(n - 1)}\right)^{\!Z}
  \I{S=0}
  = 
  \frac{1}{n^{2\ell}}\!\left(\frac{n}{n - 1}\right)^{\!Z}
  \I{S=0}.\qedhere
\]
\end{proof}

\begin{tcolorbox}
\begin{theorem}\label{thm:anti_cycle_expectation}
Let $X_{\ell}$ denote the number of directed $2\ell$-cycles in the
best-response digraph of an $n \times n$ random antisymmetric bimatrix game.
If $2 \le \ell \le n$, 
\begin{align}
\label{anti-bound}
    \E{X_{\ell}}
    \le 4 \mu_{\ell}
\end{align}
Moreover, if $\ell=o(n)$,
\[
  \E{X_{\ell}}
  = 
  \left(1 + O\!\left( \ell^2/n^2 \right)\right) \mu_{\ell}\,.
\]
For any fixed $\ell$,
\[
X_\ell\overset{d}{\to} 
\Po\lp1/\ell\rp.
\]
\end{theorem}
\end{tcolorbox}

\begin{proof}[Proof of Theorem \ref{thm:anti_cycle_expectation}]
Let $C$ be uniform on $\cC_{2\ell}$ and write $Z = Z(C)$ and $S = S(C)$.
Lemma~\ref{lem:cycle_probability} implies
\[
  \E{X_{\ell}}
  = 
  \Abs{ \calC_{2\ell} } \cdot
  \E{
    \frac{1}{n^{2\ell}}\!\left(\frac{n}{n - 1}\right)^{\!Z}\!
    \I{S = 0}
  }
  = \mu_{\ell} \cdot \E{\left(\frac{n}{n-1}\right)^{\!Z} \I{S = 0}},
\]
recalling $\mu_{\ell} = \Abs{ \calC_{2\ell} } / n^{2\ell}$.
Using Lemma~\ref{lem:overlap}, we have
\[
\E{\left(\frac{n}{n-1}\right)^{\!Z} \I{S = 0}}
\le \E{\left(\frac{n}{n-1}\right)^{\!Z} }
=1+\O{\ell^2/n^2},
\]
and the upper bound in \eqref{anti-bound} follows by $Z\le n$
and $(1+1/(n-1))^n\le 4$.

For the lower bound,
since $Z \le \ell $, we have
\[
  \left(\frac{n}{n-1}\right)^{\!Z}
  \le
  \left(\frac{n}{n-1}\right)^{\!\ell}
  \le
  \exp\!\left(\frac{\ell}{n - 1}\right)
  = 
  O(1),
\]
and then
\begin{align*}
\E{\left(\frac{n}{n-1}\right)^{\!Z} \I{S = 0}}
&=
  \E{\left(\frac{n}{n-1}\right)^{\!Z}}
-
  \E{\left(\frac{n}{n-1}\right)^{\!Z} \I{S \ge 1}} \\
&\ge
1+\O{\ell^2/n^2}
- \O{ \E{S} }
= 1+\O{\ell^2/n^2},
\end{align*}
where we use Lemma~\ref{lem:overlap}
and $\E{\I{S\ge1}}=\P{S\ge1}\le\E{S}$.
Combining the two bounds yields the asymptotics of the expected cycle count.

For the limiting distribution,
note that for vertex-disjoint cycles $C_1, \ldots, C_k \in\calC_{2\ell}$, 
a similar argument to that in Lemma \ref{lem:cycle_probability} gives
\begin{align}
\label{anti-cycles}
\P{ \bigcap_{i \in [k]} \{ C_i \subset\calD \}}
=
\frac{1}{n^{2k\ell}}\!\bR{\frac{n}{n-1}}^{Z } 
\I{ S=0 },
\end{align}
where
the common indices (between $\{ \calI(C_i): i \in [k]\}$
and $\{ \calJ(C_i): i \in [k]\}$)
and symmetric pair counts in 
$\cup_{i \in [k]} C_i$ are
\[
Z
= \abs{  \bR{ \bigcup_{i\in[k]} \calI(C_i) } \cap 
\bR{ \bigcup_{i\in[k]} \calJ(C_i) } }
\quad\mbox{and}\quad
S= S\bR{ \bigcup_{i \in [k]} C_i }.
\]
Therefore, we have
\[
\frac{1}{n^{2k\ell}}
\I{Z=0}
\le \P{ \bigcap_{i \in [k]} \{ C_i \subset\calD \}}
\le
\frac{1}{n^{2k\ell}}\!\bR{\frac{n}{n-1}}^{ k\ell },
\]
where the lower bound follows from \eqref{anti-cycles} by setting $Z=0$, also noting that $Z=0$ implies $S=0$. The upper bound follows by noting that $Z \le k\ell$.
Similarly to the independent case, the number of such ordered tuples is $(n)_{k\ell}^2/\ell^k$, 
hence, the $k$-th factorial moment of $X_\ell$ is
\begin{align*}
\E{ (X_\ell)_k }
= \frac{(n)_{k\ell}^2}{\ell^k }
\P{ \bigcap_{i \in [k]} \{ C_i \subset\calD \}}
\to \lp1/\ell\rp^k.
\end{align*}
This completes the proof by the method of moments.
\end{proof}

We next obtain the length and number of cycles in the best-response digraph in antisymmetric games.
\begin{tcolorbox}
    \begin{theorem}
Let $\calD$ be the  
best-response digraph
of an
$n \times n$
random antisymmetric bimatrix game
and let the length of the longest cycle be $2\ell_{\max}$.
Let $T_n$ denote the total number of directed cycles in $\calD$.
Then
\[
\ell_{\max}=\Theta_\mathbb{P}(\sqrt{n}),
\quad\mbox{ and }\quad
\E{T_n}=\lp\frac{1}{2}+o(1)\rp \log n.
\]
\end{theorem}
\end{tcolorbox}

\begin{proof}
Let $0<\ep<1$ be fixed,
and define
\[
c=e^{-2e/\ep}
\quad\mbox{and}\quad
C= \sqrt{\log ( 1+e^2/\ep )}.
\]
The first moment calculation is similar to the proof of Theorem \ref{thm:longest}.
By Theorem \ref{thm:anti_cycle_expectation},
the expected cycle count 
in the antisymmetric game is upper bounded by $4$ times that
in the independent case.
So the upper bound follows.

Similarly to the proof of Theorem \ref{thm:longest}, let
\[
\calT
=
\bC{
C\in \calC_{2\ell}:
\ceil{ c\sqrt n }
\le\ell\le
\ceil{ \sqrt n }
}.
\]
By
Lemma \ref{lem:cycle_probability}, we have, for $\ell\le \ceil{\sqrt n}$, that
\begin{align*}
\mu
&=
\sum_{\ell=\ceil{ c\sqrt n }}
^{\ceil{ \sqrt n }}
\sum_{C\in\calC_{2\ell}}
\E{ \frac{1}{n^{2\ell}}\!\bR{\frac{n}{n-1}}^{Z} \I{ S=0} } \\
&\ge 
\sum_{\ell= \ceil{ c\sqrt n }}^{
\ceil{ \sqrt n }}
\mu_\ell \bR{ 1 - \frac{\ell(\ell-2)}{n(n-1)} } 
\ge
\bR{ 1 - \frac{3}{n} } \sum_{\ell= \ceil{ c\sqrt n }}^{
\ceil{ \sqrt n }}
\mu_\ell
\ge
(1+o(1))
e^{-1} \log(1/c),
\end{align*}
where 
we use $\E{ \I{S=0} } = 1-\P{S\ge1}
\ge 1- \E{S}$ in combination with
Lemma \ref{lem:overlap},
and the last summation is similar to that in the independent case.

Using the probability of joint occurrences of cycles in \eqref{anti-cycles},
we have
\begin{align}
\label{anti-2cycles}
\Delta
=
\sum_{
\substack {C_i\ne C_j\in\calT \\
C_i \cap C_j=\emptyset}
}
\E{ Y_{C_i}Y_{C_j} }
&\le 
\sum_{
\substack {C_i\ne C_j\in\calT \\
C_i \cap C_j=\emptyset}
}
\frac{1}{n^{ |C_i| + |C_j| } }\!\bR{\frac{n}{n-1}}^{2 \ceil{\sqrt n }} \notag\\
&\le
\Exp{\frac{5}{\sqrt{n}}}
\bR{ \sum_{ C \in\calT}
\frac{1}{n^{ |C| } } }^2 
\le (1+o(1))\mu^2,
\end{align}
where the bound on the exponent follows by
noting that $|C| \le 2\ceil{\sqrt n}$ for any $C\in \calT$.
By choosing $\alpha=0$ and $\beta=o(1)$
in Lemma \ref{2rd-moment},
we have
\[
\P{ \ell_{\max} < c\sqrt{n}}
\le
\P{ \sum_{C\in\calT}Y_C=0} \le 
(1+o(1))/\mu + o(1)
\le \eps.
\]
This completes the proof of the first part.

To compute the expected total cycle count,
as in the proof for the independent case, we split the set $[n]\cap\mathbb{Z}$ into $I_1=[1,\sqrt{n/\log n}]$, $I_2=(\sqrt{n/\log n},\sqrt{n\log n}]$, and $I_3=(\sqrt{n\log n},n]$.
For $\ell\in I_3$, using \eqref{anti-bound}, we have
\begin{align*}
\sum_{\ell\in I_3}
\E{X_{\ell}}
\le e^2\sum_{\ell\in I_3} \frac{1}{\ell}
\Exp{ -\frac{(\ell-1)^2}{n} }
\le (1+o(1))\frac{e^2n}{ \sqrt{n\log n} }
e^{ - \log n }
\sim 
\frac{e^2}{ \sqrt{n\log n} }=o(1).
\end{align*}
For $\ell\in I_1\cup I_2$, 
we have
$\E{X_{\ell}}
=
  \left(1 + O\!\left(\ell^{2}/n^{2}\right)\right) \mu_{\ell}\sim \mu_\ell$,
and similarly to the independent case,
the proof follows.
\end{proof}

\section{Random common-interest games}

In this section, we study common-interest games,
where both players share the same random payoff matrix.
We estimate the expected number of best-response paths of length $k$, and the order of the longest directed best-response path.

\subsection{PNE and best-response paths}

Common-interest games are known to have a pure Nash equilibrium, which can be reached by the best-response dynamics.
In particular, Monderer and Shapley \cite{monderer1996potential}, showed that every finite ordinal potential game possesses a pure-strategy equilibrium. 





Next we show that
in random common-interest games, there are no cycles of length larger than $2$, 
and the best-response dynamics always lead to a pure Nash equilibrium.
\begin{tcolorbox}
\begin{lemma}
Let $\calD$ be the  
best-response digraph
of an
$n \times n$
random common-interest game. 
\begin{enumerate}
\item Along any (directed) path in $\calD$, the corresponding payoffs form a strictly increasing sequence until reaching a pure Nash equilibrium. 
\item
For each vertex in $\calD$,
there exists a unique (directed) path in $\calD$ from it
that leads to a unique pure Nash equilibrium.
\item Let
$T_n$ denote the total number of directed cycles in $\calD$.
Then
\begin{align*}
    T_n = N_{\PNE}
    \quad\mbox{ and }\quad
    \E{ N_{\PNE} } 
    = \frac{n^2}{2n-1}
    \sim \dfrac{1}{2}n\,.
\end{align*}
\end{enumerate}
\end{lemma}
\end{tcolorbox}
\begin{proof}

A directed $2$-cycle is formed
when simultaneously the maximum in
its row and in its column are reached, hence exactly a pure Nash equilibrium. 
Along any directed
walk that has not yet entered such a $2$-cycle, the payoff at the newly selected cell is
strictly larger than at the preceding cell. Thus
$T_n=N_{\PNE}$. 

In a random payoff matrix, an entry is the largest among the $2n-1$ entries in the union of its
row and column with probability $1/(2n-1)$. Then the expected PNE count follows from the linearity of expectation.
\end{proof}

\subsection{Path count and the longest path length}

In this section, a path in digraph $\calD$ has distinct vertices, and the length is the number of edges.
Note that the maximum possible path length is $2n-1$. 
For integer $j \ge 1$, let
\begin{align}
\label{Hj-def}
    H(j) := jn - \floor{ j^2/4 }.
\end{align}
This is the number of entries in the union of $j$ distinct alternating rows and columns of the $n\times n$ payoff matrix, as there are
$\ceil{ j/2 }$ rows and $\floor{ j/2 }$ columns, or vice versa,
and this yields $\floor{ j^2/4 }$ intersections.

\begin{tcolorbox}
\begin{lemma}
For $1 \le k\le2n-1$,
let 
$v_0, v_1, \ldots, v_k$
be a sequence of distinct $k+1$ vertices 
alternating between the vertex bipartition of $\calD$.
Then
the probability that 
   this fixed sequence of $k+1$ vertices forms a path of length $k$ is
\begin{eqnarray*}
    P_k 
:= \P{v_0 \to v_1 \to \cdots \to v_k }
= \prod_{i\in[k]} \frac{1}{H(i)}
= \frac{1}{n^k k!} \prod_{i\in[k]} \bR{1 - \frac{\floor{ i^2/4 }}{in}}^{-1}.
\end{eqnarray*}
Moreover, for $k=o(n)$,
\begin{align}
\label{Pk-common-int}
    P_k
=
\frac{1}{n^kk!}
\Exp{\frac{k^2}{8n}+\frac{k^3}{96n^2}
+O\lp \frac{k}n+\frac{k^4}{n^3}\rp}\,.
\end{align}
\end{lemma}
\label{path-prob}
\end{tcolorbox}

A similar result also appears in 
\cite[Eq. (7.2)]{collevecchio2024basins},
for ordinal potential games.

\begin{proof}
The distinct $k+1$ vertices form a directed path
if the $k$ successive edges all appear in $\calD$.
Note that for the common interest game, 
the payoff
corresponding to the $i$-th edge must be the largest entry in this union of
$H(i)$ entries,
in view of the definition \eqref{Hj-def},
and every entry has equal probability of being that maximum.
Taking the log of the product factor gives
\begin{align*}
\sum_{i \in [k]} \frac{\floor{ i^2/4 }}{in}
+ \frac12\sum_{i \in [k]} \bR{ \frac{
\floor{ i^2/4 }}{in} }^2
+ \O{ \frac{k^4}{n^3} }
&= 
\sum_{i \in [k]} \frac{i}{4n}
+ \sum_{i \in [k]} \frac{ i^2 }{32 n^2}
+ \O{ \frac{k}{n}+\frac{k^4}{n^3} } \\
&= 
\frac{k^2}{8n}
+ \frac{k(k+1)(2k+1)}{192n^2}
+ \O{ \frac{k}{n}+\frac{k^4}{n^3} }.
\qedhere
\end{align*}
\end{proof}

\begin{tcolorbox}
\begin{theorem}
\label{pathCount}
Let $\calD$ be the  
best-response digraph
of an
$n \times n$
random common-interest bimatrix game.
Let
$Y_k$ be the number of paths of length $k$ in $\calD$.
Then for $k=o(n)$,
\begin{align*}
    \E{ Y_k } 
    =\frac{2n}{k!}\Exp{-\frac{k^2}{8n}-\frac{k^3}{32n^2}+O\lp\frac{k}{n}+\frac{k^4}{n^3}\rp}\,.
\end{align*}
Moreover, if $k=o(n^{2/3})$, then $\E{ Y_k } \sim 2ne^{-k^2/8n}/k!$; and we have $\E{ Y_k } \sim 2n/k!$
if $k=o(\sqrt{n})$.
\end{theorem}
\end{tcolorbox}

\begin{proof}
For $1 \le k\le2n-1$.
 Choose a sequence of $k+1$ distinct vertices 
$v_0, v_1, \ldots, v_k$
alternating between the vertex bipartition.
For even $k= 2j$ with integer $j\ge1$, the number of path candidate sequences is $2(n)_{j+1}(n)_{j}$.
Therefore, 
\begin{align*}
\E{ Y_{2j} }
= 2(n)_{j+1}(n)_{j} P_{2j}
&= 2n^{k+1} \bR{1-\frac{k}{2n}}
\prod_{j=1}^{k/2-1} \bR{1 - \frac{j}{n}}^2 P_k \\
&=\frac{2n}{k!}\Exp{-\frac{k^2}{8n}-\frac{k^3}{32n^2}+O\lp\frac{k}{n}+\frac{k^4}{n^3}\rp},
\end{align*}
where we use the asymptotics of $P_k$ in \eqref{Pk-common-int} and also note that, for $k=o(n)$,
\[
\prod_{j=1}^{k/2-1} \bR{1 - \frac{j}{n}}^2
= 
\Exp{-\frac{k^2}{4n}-\frac{k^3}{24n^2}
+O\lp \frac{k}n+\frac{k^4}{n^3}\rp }.
\]
For odd $k = 2j+ 1$, the number of path candidates is $2(n)_{j+1}^2$, and a similar derivation yields the same asymptotic formula.
If $k=o(n^{2/3})$, then $k^3/n^2 = o(1)$; 
if $k=o(\sqrt n)$, then $k^2/n=o(1)$. 
The two statements follow.
\end{proof}

\begin{tcolorbox}
\begin{theorem}Let $\calD$ be the  
best-response digraph
of an $n \times n$ random common-interest bimatrix game.
For every $\ep >0$, 
the length $L_n$ of the longest path in $\calD$ satisfies
\[
\lim_{n\to \infty} \P{\Big\lfloor(1-\ep) \dfrac{\log n}{\log\log n} \Big\rfloor\le L_n \le \Big\lceil(1+\ep) \dfrac{\log n}{\log\log n}\Big\rceil} =1\,.
\]
\end{theorem}
\end{tcolorbox}

\begin{proof}
Let $0<\ep<1$ be fixed.
Define
\begin{align}
\label{mM}    
m
:= \left\lfloor (1-\ep)
\dfrac{\log n}{\log\log n}
\right\rfloor
\quad\mbox{ and }\quad
M
:= \left\lceil (1+\ep)
\dfrac{\log n}{\log\log n}
\right\rceil.
\end{align}
Recall that $Y_k$ denotes the number of paths of length $k$.
By Theorem \ref{pathCount},
\[
\E{ Y_{M} }
= \frac{2n}{M!}\Exp{-\frac{M^2}{8n}-\frac{M^3}{32n^2}+O\lp
\frac{M}{n}+\frac{M^4}{n^3}\rp}
\sim \frac{2n}{M!}.
\]
Stirling's formula implies
\[
M! 
\sim \sqrt{2\pi M} \bR{\frac{M}{e}}^{M}
\ge n^{1+\ep+o(1)},
\]
where
we use 
$- M + M \log M 
\sim M \log M 
\sim (1+\ep) \log n$.
Every directed path of length at least $M$ contains a directed prefix of length $M$.
By Markov's inequality, 
\[
\P{ L_n \ge M}
\le \P{ Y_{M} \ge 1 }
\le \E{ Y_M }
\le 2n^{-\ep+o(1)}
\to 0.
\]

To prove the lower bound, 
it is useful to define 
$\calP_m$ as the
set of all paths of length $m$ in $K_{n, n}$.
By the second moment method (Lemma \ref{2rd-moment})
\begin{align}
\label{lower-comm}
\P{Y_{m}=0}
\leq \frac{ 1+\alpha}{\mu}+\beta,
\end{align}
where, with some $\alpha, \beta\ge0$ that we choose later,
\begin{align*}
\mu=\E{Y_{m}},
\quad
\Delta=\sum_{P\ne Q \in 
\calP_m}
\P{ \calE_P
\cap
\calE_Q }
\le \alpha\mu+(1+\beta)\mu^2,
~~\mbox{ with }~~
\calE_P
:= \{ P \subset \calD \}
\end{align*}
for any path $P \in\calP_m$,
and $\Delta$
is a sum over ordered pairs of distinct paths.
Similarly to the upper bound above,
it is easy to verify that
$\E{ Y_m} \to \infty$.
Fix an arbitrary path $P \in \calP_m$.
Next, we show that
\begin{align}
\delta_1
=
\delta_1(P)
&:=\sum_{Q \in 
\calP_m: P \cap Q = \emptyset}
\P{ \calE_P \cap \calE_Q }
\le (1+o(1)) \mu P_m,  \label{delta1}\\
\delta_2
=\delta_2(P)
&:=\sum_{Q \in 
\calP_m \setminus \{P\}, P \cap Q \ne \emptyset}
\P{ \calE_P \cap \calE_Q }
= O(P_m).
\label{delta2}
\end{align}
This completes the proof of the lower bound in view of \eqref{lower-comm} by 
choosing
$\alpha=O(1)$ and $\beta=o(1)$, as
\begin{align*}
\Delta\le
\sum_{P \in \calP_m} 
(\delta_1+\delta_2) 
\le (1+o(1))\mu^2
+ \O{\mu}.
\end{align*}

Hence, it remains to prove \eqref{delta1}
and \eqref{delta2}.\\
\noindent
\textbf{Proof of \eqref{delta1}.}
Note that for vertex-disjoint paths
$P$ and $Q$,
events $\calE_P$ and $\calE_Q$
are dependent if 
a row (column) in $P$ intersects a column (row) in $Q$.
Fix two vertex-disjoint paths
$P$ and $Q$. Below we define the event
$\calE_Q'$ 
with all the intersecting entries removed from the payoff matrix.
Note that 
the number of deletions is upper bounded 
by $2m^2$, and the new event
$\calE_Q'$
is independent of $\calE_P$,
see Figure \ref{fig:EQprime} for illustration.
Formally,
let $\calI(P)$ and $\calJ(P)$ be the rows and columns used by any given path $P$.
Denote the path $Q$ in $\calD$ of length $m$ by 
$Q=(q_0, q_1, q_2, \ldots, q_{m})$.
Let
\[
\calL(r_i) := \{i\} \times [n],
~
\calL(c_j) := [n] \times \{j\}
~~\mbox{and}~
\calR=\calR(P,Q) := (\calI(P)\times\calJ(Q)) \cup 
(\calI(Q)\times\calJ(P)),
\]
and define
\[
\calE_Q' := 
\bigcap_{t=0}^{m-1}\bC{ 
a_{ \calL(q_t) \cap \calL(q_{t+1}) }
=
\max_{ (i,j) \in \calL(q_t) \setminus \calR
} a_{i,j} },
~\mbox{and}~
\calE_Q := 
\bigcap_{t=0}^{m-1}
\bC{ 
a_{ \calL(q_t) \cap \calL(q_{t+1}) }
=
\max_{ (i,j) \in \calL(q_t)
} a_{i,j} }.
\]
Then $\mathcal E_Q\subseteq\mathcal E_Q'$.

\begin{figure}[htbp]
\centering
\begin{subfigure}[b]{0.48\textwidth}
\centering
\begin{tikzpicture}[x=0.8cm,y=0.8cm,line width=0.8pt]
    \foreach \r/\c in {1/3,1/4,2/3,2/4,3/1,3/2,4/1,4/2}
        \fill[pattern=north east lines,pattern color=black!45]
            (\c-1,-\r) rectangle (\c,-\r+1);

    \foreach \i in {1,...,4} {
        \draw[dashed,line width=0.4pt,black!55] (\i,0) -- (\i,-5);
        \draw[dashed,line width=0.4pt,black!55] (0,-\i) -- (5,-\i);
    }
    \draw (0,0) rectangle (5,-5);

    \foreach \i in {1,...,5} {
        \node[font=\scriptsize] at (-0.3,-\i+0.5) {$\i$};
        \node[font=\scriptsize] at (\i-0.5,0.3) {$\i$};
    }

    \fill (0.5,-0.5) circle (3.5pt);
    \fill (0.5,-1.5) circle (3.5pt);
    \fill (1.5,-1.5) circle (3.5pt);
    \draw[-{Stealth[length=2.2mm]},shorten <=4pt,shorten >=4pt]
        (0.5,-0.5) -- (0.5,-1.5);
    \draw[-{Stealth[length=2.2mm]},shorten <=4pt,shorten >=4pt]
        (0.5,-1.5) -- (1.5,-1.5);

    \draw[fill=white] (2.5,-2.5) circle (3.5pt);
    \draw[fill=white] (2.5,-3.5) circle (3.5pt);
    \draw[fill=white] (3.5,-3.5) circle (3.5pt);
    \draw[-{Stealth[length=2.2mm]},shorten <=4pt,shorten >=4pt]
        (2.5,-2.5) -- (2.5,-3.5);
    \draw[-{Stealth[length=2.2mm]},shorten <=4pt,shorten >=4pt]
        (2.5,-3.5) -- (3.5,-3.5);
\end{tikzpicture}
\end{subfigure}
\hfill
\begin{subfigure}[b]{0.48\textwidth}
\centering
\begin{tikzpicture}[x=1cm,y=0.78cm,
    brv/.style={circle,draw,line width=0.8pt,fill=white,
                inner sep=0pt,minimum size=6.4mm},
    rarr/.style={-{Stealth[length=2.6mm]},line width=0.8pt,
                 shorten >=0.5pt,shorten <=0.5pt},
    carr/.style={-{Stealth[length=2.6mm]},line width=0.8pt,
                 densely dashed,shorten >=0.5pt,shorten <=0.5pt}
]
    \foreach \i in {1,...,5} {
        \node[brv] (r\i) at (0,{0.5-\i}) {$r_{\i}$};
        \node[brv] (c\i) at (3.6,{0.5-\i}) {$c_{\i}$};
    }

    \draw[rarr] (c1) -- (r2);
    \draw[rarr] (r2) -- (c2);

    \draw[carr] (c3) -- (r4);
    \draw[carr] (r4) -- (c4);
\end{tikzpicture}
\end{subfigure}
\caption{Two vertex-disjoint paths $P=(c_1,r_2,c_2)$ and $Q=(c_3,r_4,c_4)$ of length $2$.
Solid and hollow dots correspond to the payoffs of $P$ and $Q$, respectively.
The shaded entries are the deleted entries, 
that are the intersecting entries of a row of $P$ and a column of $Q$, or vice versa.
The event $\mathcal E_Q$ requires each hollow dot to be the maximum of all entries in a row or a column,
while $\mathcal E_Q'$ only requires it to be the maximum of the unshaded entries of its row or column. Hence $\mathcal E_Q\subseteq\mathcal E_Q'$, and $\mathcal E_Q'$ depends only on entries that are disjoint from the entries on rows or columns of $P$, 
so it is independent of $\mathcal E_P$.}
\label{fig:EQprime}
\end{figure}
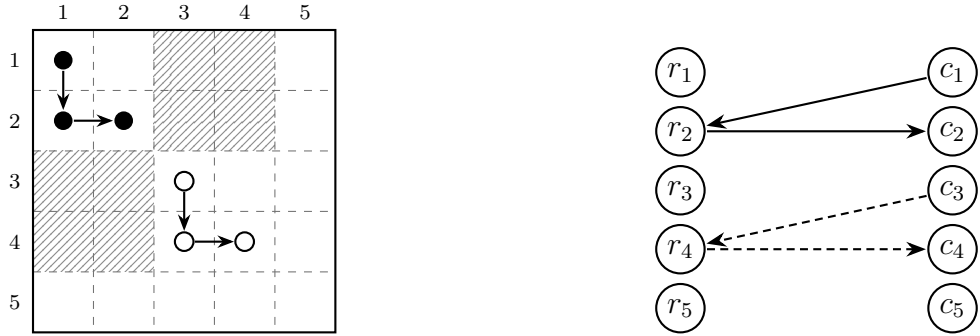
Similarly to Lemma \ref{path-prob},
we have
\[
\P{ \calE_Q' }
\le
\prod_{i\in[m]} \frac{1}{in-\floor{i^2/4}-2m^2}
\sim 
\prod_{i\in[m]} \frac{1}{in-\floor{i^2/4}}
=
\P{ \calE_Q },
\]
where the asymptotic equivalence
follows from the definition of $m$
in \eqref{mM}.
Hence,
we have
\begin{align*}
\delta_1
&=
\P{ \calE_P }
\sum_{Q \in 
\calP_m: P \cap Q = \emptyset}
\P{ \calE_Q \mid \calE_P } \\
&\le
\P{ \calE_P }
\sum_{Q \in 
\calP_m: P \cap Q = \emptyset}
\P{ \calE_Q' \mid \calE_P }
\sim 
\sum_{Q \in 
\calP_m: P \cap Q = \emptyset}
\P{ \calE_P }^2
\le \abs{ \calP_m }
\P{ \calE_P }^2
= \mu P_m,
\end{align*}
where the first inequality 
holds because 
the arc is more likely to appear as the intersecting entries are deleted.
\\
\noindent
\textbf{Proof of \eqref{delta2}.}
For two distinct intersecting paths $P, Q \in \calP_m$,
suppose $P$ and $Q$ meet at vertex $u=P_a=Q_b$ 
with $0\le a, b \le m$
and $a+b>0$,
as in Figure \ref{fig:2p}.
\begin{figure}[htb]
\centering
\resizebox{0.9\linewidth}{!}{%
\begin{tikzpicture}[
  Pline/.style={-{Stealth[length=5pt]}, thick},
  Qline/.style={-{Stealth[length=5pt]}, thick, dashed},
  vert/.style={circle, fill=black, inner sep=1.5pt},
  lbl/.style={font=\normalsize},
]

\node[vert] (P0)   at (0,0)    {};
\node[vert] (P1)   at (1.8,0)  {};
\node[vert] (Pam1) at (4.6,0)  {};
\node[vert] (u)    at (6.6,0)  {};
\node[vert] (Pap1) at (8.6,0)  {};
\node[vert] (Pmm1) at (13.4,0) {};
\node[vert] (Pm)   at (15.4,0) {};

\node[lbl, below=3pt of P0]   {$P_0$};
\node[lbl, above=3pt of P1]   {$P_1$};
\node[lbl, above=3pt of Pam1] {$P_{a-1}$};
\node[lbl, above=3pt of u]    {$u=P_a=Q_b$};
\node[lbl, below=3pt of Pap1] {$P_{a+1}{=}Q_{b+1}$};
\node[lbl, above=3pt of Pmm1] {$P_{m-1}{=}Q_{b+m-a-1}$};
\node[lbl, below=4pt of Pm]   {$P_m=Q_{b+m-a}$};

\draw[Pline] (P0)   -- (P1);
\draw[Pline] (P1)   -- node[lbl, above] {$\cdots$} (Pam1);
\draw[Pline] (Pam1) -- (u);

\draw[Pline, shorten >=3pt, shorten <=3pt] ($(u)+(0,0.10)$)    -- ($(Pap1)+(0,0.10)$);
\draw[Qline, shorten >=3pt, shorten <=3pt] ($(u)+(0,-0.10)$)   -- ($(Pap1)+(0,-0.10)$);
\draw[Pline, shorten >=3pt, shorten <=3pt] ($(Pap1)+(0,0.10)$) -- ($(Pmm1)+(0,0.10)$);
\draw[Qline, shorten >=3pt, shorten <=3pt] ($(Pap1)+(0,-0.10)$) -- ($(Pmm1)+(0,-0.10)$);
\node[lbl] at ($(Pap1)!0.5!(Pmm1)$) [above=1pt] {$\cdots$};
\draw[Pline, shorten >=3pt, shorten <=3pt] ($(Pmm1)+(0,0.10)$)  -- ($(Pm)+(0,0.10)$);
\draw[Qline, shorten >=3pt, shorten <=3pt] ($(Pmm1)+(0,-0.10)$) -- ($(Pm)+(0,-0.10)$);

\node[vert] (Q0)   at (0.6,-2.0) {};
\node[vert] (Q1)   at (2.4,-2.0) {};
\node[vert] (Qbm1) at (5.6,-2.0) {};
\node[lbl, below=3pt of Q0]   {$Q_0$};
\node[lbl, below=3pt of Q1]   {$Q_1$};
\node[lbl, below=3pt of Qbm1] {$Q_{b-1}$};

\draw[Qline] (Q0) -- (Q1);
\draw[Qline] (Q1) -- node[lbl, below] {$\cdots$} (Qbm1);
\draw[Qline] (Qbm1) to[out=45, in=225] (u);

\node[vert] (Qs1)  at (16.4,2.0) {};
\node[vert] (Qmm1) at (21.2,2.0) {};
\node[vert] (Qm)   at (23.0,2.0) {};
\node[lbl, above=3pt of Qs1]  {$Q_{b+m-a+1}$};
\node[lbl, above=3pt of Qmm1] {$Q_{m-1}$};
\node[lbl, above=3pt of Qm]   {$Q_m$};

\draw[Qline] (Pm)  to[out=45, in=225] (Qs1);
\draw[Qline] (Qs1) -- node[lbl, above] {$\cdots$} (Qmm1);
\draw[Qline] (Qmm1) -- (Qm);

\end{tikzpicture}%
}
\caption{Distinct paths $P$ and $Q$ of length $m$ intersect at vertex 
$u=P_a=Q_b$.}
\label{fig:2p}
\end{figure}
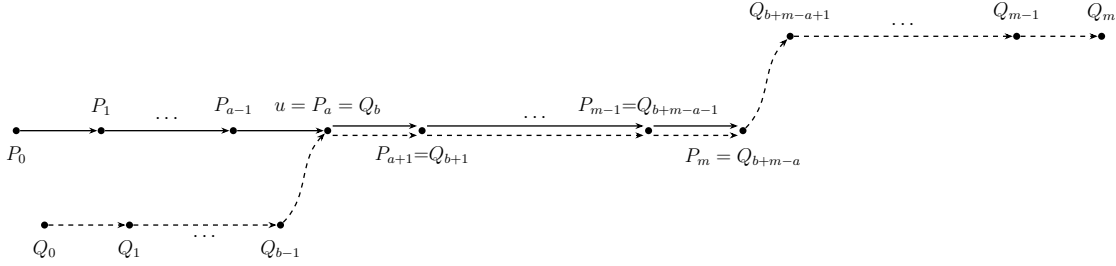

We show that
\begin{align}
\label{3parts}
\P{ \calE_Q \mid \calE_P }
\le (1+o(1)) P_b P_{(a-b)_+},
\end{align}
where $(a-b)_+=\max(a-b,0)$,
and $P_{(a-b)_+}=1$ if $b\ge a$.
Then
$\P{ \calE_Q \cap \calE_P }
\le
(1+o(1))
P_m
P_b P_{(a-b)_+}$,
and therefore, by Lemma \ref{path-prob},
\begin{align*}
\delta_2
\le
(1+o(1)) P_m
\sum_{Q \in 
\calP_m \setminus \{P\}, P \cap Q \ne \emptyset}
P_b P_{(a-b)_+}
&\sim 
P_m
\sum_{Q \in 
\calP_m \setminus \{P\}, P \cap Q \ne \emptyset}
\frac{1}{n^b b!}
\frac{1}{n^{(a-b)_+} (a-b)_+!} \\
&\le
P_m
\sum_{a,b\ge0}
\frac{1}{b!}
\frac{1}{(a-b)_+!}
= O(P_m).
\end{align*}
It remains to prove \eqref{3parts}.
Note that 
conditioning on the occurrence of $P$, 
in order to have $Q$ in $\calD$,
we need the two subpaths to appear: the prefix subpath of $Q$ before the merging vertex $u$, and the suffix subpath after the last vertex $P_m$ of $P$, that is,
\begin{align*}
\P{ \calE_Q \mid \calE_P }
= 
\P{ \calE_{Q_0\ldots Q_{b} } 
\mid \calE_P }
\P{
\calE_{Q_{b+m-a}\ldots Q_{m}} 
\mid \calE_P, 
\calE_{Q_0\ldots Q_{b} }}.
\end{align*}
Note that the vertices $P_{a-1}$ and
$Q_{b-a}$ 
before the merging vertex 
$u$ must be on the same row or column in the payoff matrix in order to have the merging paths.
Conditioning on the occurrence of arc $P_{a-1}P_a$ in $\calD$, arc
$Q_{b-1}Q_b$ is more likely to appear, so we have
\begin{align*}
\P{ \calE_{Q_0\ldots Q_{b} } 
\mid \calE_P }
\le 
(1+o(1)) P_b,
\end{align*}
where we note that 
the subpath
$Q_0Q_1\cdots Q_{b-1}$ is vertex-disjoint from $P$,
so it is similar to the disjoint case.
Combining with 
a similar argument for 
$\P{
\calE_{Q_{b+m-a}\ldots Q_{m}} 
\mid \calE_P, 
\calE_{Q_0\ldots Q_{b} }}$ yields \eqref{3parts}.
The proof is now complete.
\end{proof}

\section{Random zero-sum games}
In this section we study random zero-sum games, 
where
the two random payoff matrices satisfy $\bm A+\bm B=\mathbf{0}$.
We estimate the expected number of cycles of a fixed length, the expected total number of cycles, and the typical length of the longest cycle.

The case $\ell=1$ 
corresponds to the PNE,
and it is exponentially rare.
By Goldman \cite{Goldman1957ThePO},
we have
$
\E{ N_\PNE } 
= (n!)^2 / (2n-1)! $.
Using $\binom{2n}{n}\sim 4^n/\sqrt{\pi n}$
gives the following asymptotics:

\begin{tcolorbox}
\begin{proposition}
[Goldman \cite{Goldman1957ThePO}]
\label{prop:saddle}
\begin{align}
\label{eq:saddle}
\E{ N_\PNE } 
= \frac{ (n!)^2 }{ (2n-1)! }
= \frac{2n}{\binom{2n}{n}} \sim 2\sqrt{\pi}\, n^{3/2} 4^{-n}~.
\end{align}
\end{proposition}
\end{tcolorbox}


For $\ell \ge 2$, 
a $2\ell$-cycle 
$
C = (r_{i_1}, c_{j_1}, \ldots,
r_{i_k}, c_{j_k}, \ldots,
r_{i_\ell}, c_{j_\ell}, r_{i_1})
$
appears in $\calD$
if and only if
the following $2\ell$ events hold
\[
\bC{
b_{i_k,j_k}
=\textstyle\displaystyle\max_{m\in[n]}b_{i_k,m}
}_{k=1}^{\ell}\quad\text{and}\quad
\bC{
a_{i_{k+1},j_k}
=\textstyle\displaystyle\max_{m\in[n]}a_{m,j_k}
}_{k=1}^{\ell}.
\]
For zero-sum games, it is equivalent to
a cycle of alternating extrema
in the payoff matrix $\bm A$.

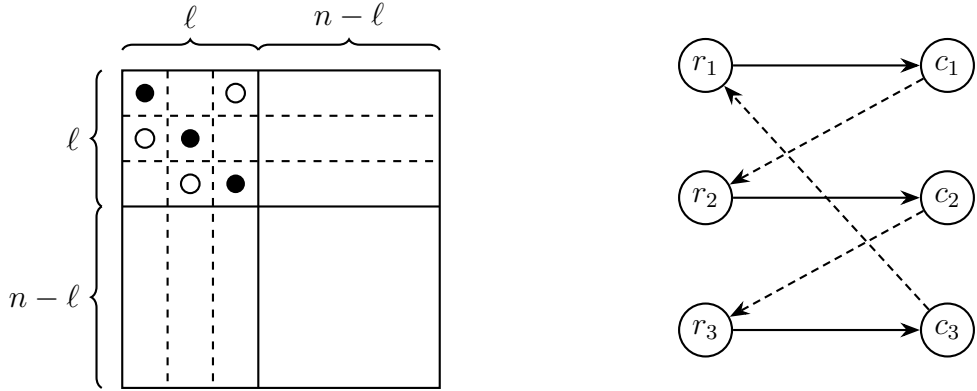
\begin{figure}[htb]
\centering
\begin{subfigure}[b]{0.48\textwidth}
\centering

\begin{tikzpicture}[x=.6cm,y=.6cm,line width=0.8pt]

\def\LL{3}\def\NN{7}
\pgfmathtruncatemacro{\Lm}{\LL-1}
\draw (0,0) rectangle (\NN,-\NN);
\foreach \i in {1,...,\Lm}{
  \draw[dashed] (\i,0) -- (\i,-\NN);
  \draw[dashed] (0,-\i) -- (\NN,-\i);}
\draw (\LL,0) -- (\LL,-\NN);
\draw (0,-\LL) -- (\NN,-\LL);

\foreach \i in {1,...,\LL}{\fill (\i-0.5,-\i+0.5) circle (3.5pt);}

\foreach \j in {1,...,\LL}{
  \pgfmathtruncatemacro{\rw}{mod(\j,\LL)+1}
  \draw (\j-0.5,-\rw+0.5) circle (3.5pt);}
\draw[decorate,decoration={brace,amplitude=4pt}]
  (0,0.45) -- (\LL,0.45) node[midway,above=4pt] {$\ell$};
\draw[decorate,decoration={brace,amplitude=4pt}]
  (\LL,0.45) -- (\NN,0.45) node[midway,above=4pt] {$n-\ell$};
\draw[decorate,decoration={brace,amplitude=4pt,mirror}]
  (-0.45,0) -- (-0.45,-\LL) node[midway,left=4pt] {$\ell$};
\draw[decorate,decoration={brace,amplitude=4pt,mirror}]
  (-0.45,-\LL) -- (-0.45,-\NN) node[midway,left=4pt] {$n-\ell$};
\end{tikzpicture}
\end{subfigure}
\begin{subfigure}[b]{0.48\textwidth}
\centering
\begin{tikzpicture}[x=1cm,y=1cm,line width=0.8pt,
  vtx/.style={circle,draw,inner sep=0pt,minimum size=7mm},
  rarr/.style={-{Stealth[length=2.6mm]}},
  carr/.style={-{Stealth[length=2.6mm]},densely dashed}]
  
\def\LL{3}\def\dy{1.75}\def\dx{3.2}

\foreach \i in {1,...,\LL}{
  \node[vtx] (r\i) at (0,-\dy*\i+\dy) {$r_{\i}$};
  \node[vtx] (c\i) at (\dx,-\dy*\i+\dy) {$c_{\i}$};}

\foreach \i in {1,...,\LL}{\draw[rarr] (r\i) -- (c\i);}

\foreach \j in {1,...,2}{
  \pgfmathtruncatemacro{\nx}{\j+1}
  \draw[carr] (c\j) -- (r\nx);}
\draw[carr] (c\LL) -- (r1);

\path (0,0.78); \path (0,-4.28);
\end{tikzpicture}
\end{subfigure}
\caption{
In the payoff matrix $\bm A$ of the row player, row minima are solid circles and column maxima are hollow circles.
The corresponding $6$-cycle in $\calD$ is shown on the right-hand side.}
\label{fig:block-digraph}
\end{figure}

\begin{tcolorbox}
    \begin{theorem}
\label{0sumThm}
Let $\calD$ be the  
best-response digraph
of an
$n \times n$
random zero-sum bimatrix game. For $2\le\ell\le n$, $n\ge3$ and any $C \in \calC_{2\ell}$,
\begin{align}
\label{0sum-cycle-prob}
\frac{1}{(n-1)^{2\ell}}\lp1-\frac{\ell^2}{(n-2)^2}- \frac{8n^2\ell^2}{4^n}\rp
\le 
\P{C \subset \calD}
\le \frac{1}{(n-1)^{2\ell}}.
\end{align}
Let $X_\ell$ be the number of $2\ell$-cycles in $\calD$.
Then for all $2\le\ell\le n$,
\begin{align}
\label{0sum-cycle-upper}
\E{X_\ell}
\le \mu_{\ell}e^{2\ell/(n-1)}
\le e^3 \mu_{\ell}.
\end{align}
Moreover, $\E{X_\ell}
\sim
\mu_{\ell}
$ if $\ell=o(n)$.
\end{theorem}
\end{tcolorbox}

\begin{proof}
Let $n\times n$ random payoff matrix $\bm A$
filled with independent uniform random variables on $[0, 1]$. 
For simplicity, we relabel indices so that a $2\ell$-cycle uses rows $1,\ldots,\ell$ and columns $1,\ldots,\ell$. We define
$p_{n,\ell}$ to be the probability that 
the minimum in row $i$ is at $(i, i)$ for $i \in [\ell]$, and the maximum in column $j$ is at $ (\sigma(j), j)$ for $j \in [\ell]$, where
$$
\sigma(j) = j \bmod \ell + 1;
$$
or equivalently,
\[
p_{n,\ell}
:=\P{\bigcap_{i=1}^{\ell}\left(
\bC{a_{ii}=\min_{m\in[n]}a_{im}}
\cap\bC{a_{\sigma(i),i}=\max_{m\in[n]}a_{mi}}
\right)}
=\P{C\subset\calD}.
\]
See Figure \ref{fig:block-digraph} for an illustration.

Denote the values in the minimum cells by $\bm y=\{ y_i: i \in [\ell]\}$, and the values in the maximum cells by $
\bm 1-\bm z
= \{ 1-z_j : j \in [\ell]\}$. 
For any $i, j\le\ell$, the entry $a_{i,j}$ in cell $(i, j)$ must be between the minimum of its row and the maximum of its column,
$y_i \le a_{ij}\le1 - z_j$.
The two extrema occupy distinct cells, so almost surely $y_i < 1-z_j$. All these inequalities together specify the 
possible payoffs necessary for the occurrence of the cycle,
$$
\calU_{\ell} :=\left\{(\bm y, \bm z)\in[0, 1]^{2\ell} : \max_{i\in[\ell]}
y_i + \max_{j\in[\ell]} z_j < 1\right\}.
$$

Define the index set of free cells by
$$\calF_{\ell} := \{(i, j) \in [\ell]^2 : i \neq j, i \neq \sigma(j)\},
$$
that is, the set of cells in the upper-left $\ell\times\ell$ square that are neither minima nor maxima. 
Then
\begin{align*}
p_{n,\ell}
&=
\int_{\calU_{\ell}}
\bR{ \prod_{ i \in [\ell] } (1-y_i) \cdot \prod_{ j \in [\ell] } (1-z_j) 
}^{n-\ell}
\prod_{(i,j)\in \calF_{\ell}} (1-y_i-z_j)
d\bm y d\bm z\\
&=\int_{\calU_{\ell}}   
\bR{ \prod_{ i \in [\ell] } (1-y_i) \cdot \prod_{ j \in [\ell] } (1-z_j) 
}^{n-2} \prod_{(i,j)\in \calF_{\ell}} (1-q(y_i)q(z_j))
d\bm y d\bm z,
\end{align*}
where we use
\begin{align*}
    1-y-z=(1-y)(1-z)\lp 1-q(y)q(z)\rp
\quad\mbox{with}\quad
q(t):=\frac{t}{1-t}
\end{align*}
defined for $t\in[0,1)$.
Note that every row and column of $E$ contains $\ell-2$ free cells, so we get $\ell-2$ extra factors $(1-y_i)$ for each $i$ and $\ell-2$ extra factors $(1-z_j)$ for each $j$.

Let independent random vectors $\bm Y =  ( Y_i : i \in [\ell] )$
and $\bm Z = ( Z_j : j \in [\ell] )$ 
have independent coordinates 
with each
defined to be the minimum of $n-1$ independent uniform variables on $[0,1]$, that is, each has probability density, for $0 < t < 1$,
$$
f_n(t) = (n-1)(1-t)^{n-2}.
$$
Then the joint density of $(\bm Y, \bm Z)$ is
\[
(n-1)^{2\ell}
\bR{ \prod_{ i \in [\ell] } (1-y_i) \cdot \prod_{ j \in [\ell] } (1-z_j) 
}^{n-2}
\]
and we have
\begin{align}\label{new pnl}
    p_{n,\ell}
    =\frac{1}{(n-1)^{2\ell}}\E{\I{(\bm Y, \bm Z)\in \calU_{\ell}}
    \prod_{(i,j)\in \calF_{\ell}}\lp 1-q(Y_i)q(Z_j)\rp}.
\end{align}

\noindent
\textbf{Upper bound.} 
On $\calU_{\ell}$, we have $y_i+z_j<1$, for every $i,j$, so every factor $1-q(Y_i)q(Z_j)$ of the product in \eqref{new pnl} is in $(0,1]$.
Then the expectation is at most $1$ in \eqref{new pnl}, and therefore, 
$p_{n,\ell}\le (n-1)^{-2\ell}$.

\noindent
\textbf{Lower bound.} 
We use the inequality $\prod_k(1-a_k) \ge 1-\sum_k a_k$ for $0 \le a_k\le 1$, and get for all $(\bm Y, \bm Z)$,
\begin{align*}
\I{(\bm Y, \bm Z)\in \calU_{\ell}}\prod_{(i,j)\in \calF_{\ell}}\lp1-q(Y_i)q(Z_j)\rp
&\ge
\I{(\bm Y, \bm Z)\in \calU_{\ell}}
\bR{ 1-\sum_{(i,j)\in \calF_{\ell}}q(Y_i)q(Z_j) } \\
&= 1-\sum_{(i,j)\in \calF_{\ell}}q(Y_i)q(Z_j)-\I{(\bm Y, \bm Z)\notin \calU_{\ell}},
\end{align*}
noting that outside $\calU_{\ell}$, 
the last term is negative,
so the above still holds.
Taking expectations, by independence,
we have
\[
\E{q(Y_i)q(Z_j)}
=\E{q(Y_i)}\E{q(Z_j)}
= \bR{(n-1)\int_0^1t(1-t)^{n-3}dt}^2
=(n-2)^{-2}
\] 
for each of the $|\calF_{\ell}|=\ell(\ell-2)$ cells,
and then by \eqref{new pnl},
\[
p_{n,\ell}\ge\frac{1}{(n-1)^{2\ell}}\lp1-\frac{\ell(\ell-2)}{(n-2)^2}
-\P{(\bm Y, \bm Z)\notin \calU_{\ell}}
\rp.
\]
The complement of $\calU_{\ell}$ is the union of the $\ell^2$ regions $\{y_i+z_j \ge 1\}$,
so by the union bound, 
\begin{align*}
    \P{(\bm Y, \bm Z)\notin \calU_{\ell}}
    &\le \ell^2 \cdot \P{Y_1+Z_1\ge 1} \notag\\
    &\le 
\ell^2 (n-1)^2\int_{
(y,z)\in [0,1]^2:
y+z\ge1 }
\bR{ (1-y)(1-z) }^{n-2}dy dz \notag\\
    &\le
    \frac{1}{2} \ell^2 (n-1)^2\cdot4^{-(n-2)},
\end{align*}
where for $(\bm y, \bm z) \in \calU_{\ell}^c$,
we use
$(1-y_i)(1-z_j)\le
1/4$, and the area of $\{(y, z) \in [0, 1]^2 : y+z \ge 1\}$ is $1/2$.
Combining upper and lower bounds finishes the proof of the first part.

Note that 
$\E{ X_\ell }
= \Abs{ \calC_{2\ell} }
\cdot \P{ C \subset\calD }$.
Using $1\le\lp\frac{n}{n-1}\rp^{2\ell}\le e^{2\ell/(n-1)}$, 
we have $\E{X_\ell}\sim\mu_\ell$ for $2\le \ell=o(n)$.
\end{proof}

\begin{tcolorbox}
\begin{theorem}
\label{zero-sum asympt distr of 2ell cycles}
In a
random zero-sum bimatrix game,
we have the following.
\begin{enumerate}
    \item For any fixed $\ell\ge2$, 
    $X_\ell\overset{d}{\to} \Po\lp1/\ell\rp$;
    \item the length of the longest cycle is
$
2\ell_{\max}=\Theta_\mathbb{P}\lp\sqrt{ n}\rp;
$
\item 
the total number of directed cycles $T_n$ in $\calD$ satisfies
$
    \E{ T_n } =\lp\frac{1}{2}+o(1)\rp  \log n\,.
$
\end{enumerate}



\end{theorem}
\end{tcolorbox}

\begin{proof}
\textbf{Part 1}.
For vertex-disjoint cycles $C_1, \ldots, C_k \in\calC_{2\ell}$,
following a similar proof of Theorem
\ref{0sumThm}, we have
\begin{align}
\label{0sum-cycles}
\P{ \bigcap_{i \in [k]} \{ C_i \subset\calD \}}
\sim (n-1)^{-2k\ell}.
\end{align}
Since only vertex-disjoint cycles contribute and there are $(n)_{k\ell}^2/\ell^k$ choices,
\[
\E{ (X_{\ell})_k }
\sim \frac{(n)_{k\ell}^2}{\ell^k}(n-1)^{-2k\ell}
\sim (1/\ell)^k,
\]
which yields the Poisson convergence.

\textbf{Part 2}.
Let $0<\ep<1$ be fixed.
Define $c= e^{-2e/\ep}
\quad\mbox{and}\quad 
C= \sqrt{\log ( 1+e^3/\ep )}$.
As in the proof of Theorem \ref{thm:longest}, it suffices to show that, as $n\to\infty$, $\P{c\sqrt{n}\leq \ell_{\text{max}}\le C\sqrt{n}}\ge 1-2\varepsilon.$

The first moment calculation is similar to the proof of Theorem \ref{thm:longest}. By Theorem \ref{0sumThm}, we have that the expected cycle count in the zero-sum game is upper bounded by $e^3$ times that in the independent case, and the upper bound follows.

Similarly to the proof of Theorem \ref{thm:longest}, let
\[
\calT
=
\bC{
C\in \calC_{2\ell}:
\ceil{ c\sqrt n }
\le\ell\le
\ceil{ \sqrt n }
}
\quad\mbox{and}\quad
\mu=\sum_{C\in\calT}\E{Y_C}.
\]
Bounding the joint occurrence of two vertex-disjoint cycles, similarly to \eqref{0sum-cycle-prob}, we have
\begin{align*}
\Delta
=
\sum_{
C_i \cap C_j=\emptyset}
\E{ Y_{C_i}Y_{C_j} }
&\le 
\sum_{
C_i \cap C_j=\emptyset}
\frac1{ (n-1)^{|C_i|+|C_j|} } \\
&\le
\bR{\frac{n}{n-1}}^{4\ceil{\sqrt n}}
\sum_{
C_i \cap C_j=\emptyset}
\frac1{ n^{|C_i|+|C_j|} } 
\le 
e^{5/\sqrt n}
\mu^2 = (1+o(1)) \mu^2,
\end{align*}
where the last inequality is similar to that in \eqref{anti-2cycles}.
Then by choosing $\alpha=0$ and $\beta=o(1)$
in Lemma \ref{2rd-moment},
we have
\[
\P{ \ell_{\max} < c\sqrt{n}}
\le
\P{ \sum_{C\in\calT}Y_C=0} \le 
(1+o(1))/\mu
\le \eps,
\]
where noting that $\ell \le \ceil{\sqrt n}$, we have
\begin{align*}
\mu 
\ge 
\sum_{\ell=\ceil{ c\sqrt n }}
^{\ceil{ \sqrt n }}
\sum_{C\in\calC_{2\ell}}
\frac{1}{(n-1)^{2\ell}}\lp1-\frac{\ell^2}{(n-2)^2}- \frac{8n^2\ell^2}{4^n}\rp
&\ge (1-o(1)) \sum_{\ell= \ceil{ c\sqrt n }}^{
\ceil{ \sqrt n }}
\frac1\ell e^{-\ell^2/n} \\
&\ge (1-o(1))e^{-1}\log(1/c).
\end{align*}
This completes the proof of Part 2.

\textbf{Part 3}.
Define $I_1, I_2, I_3$ as in the proof of Theorem \ref{number of cycles independent case}.
By \eqref{eq:saddle} in Proposition \ref{prop:saddle}, we have $\E{ N_\PNE }=o(1)$, then
\[
\E{T_n} = \E{N_{\PNE}} + \sum_{\ell\ge2} \E{X_\ell}
\sim \sum_{\ell\in (I_1\cup I_2)\setminus\{1\}} \E{X_\ell}
+\sum_{\ell\in I_3} \E{X_\ell}.
\]
For $2\le\ell\le\sqrt{n\log n}$, 
Theorem \ref{0sumThm} gives
\[
\lp1-\frac{n\log n}{(n-2)^2}-\frac{8n^3\log n}{4^n}\rp\mu_\ell\le\E{X_\ell}\le e^{2\sqrt{n\log n}/(n-1)}\mu_\ell,
\]
so $\E{X_\ell}=(1+o(1))\mu_\ell$ uniformly. 
The proof of Theorem \ref{number of cycles independent case} gives $\sum_{\ell\in I_1\cup I_2}\mu_\ell=(\frac12+o(1))\log n$ and $\sum_{\ell\in I_3}\mu_\ell=o(1)$.
Noting, by \eqref{0sum-cycle-upper}, that $\E{X_\ell}\le e^3\mu_\ell$ on $I_3$ completes the proof.
\end{proof}

    

\section{Concluding remarks}
We study the non-converging behavior of the best-response dynamics in two-player bimatrix games where payoffs are drawn from continuous distributions. 
By introducing structural constraints on the payoff matrices, specifically, symmetry, anti-symmetry, common-interest, zero-sum properties, and independence, 
we gain insight into how different game classes affect the trajectory and convergence behavior of best-response dynamics.
We provide results on the component structure of the best-response digraphs. 
In particular, our findings include results on the number of short cycles and their asymptotic distribution, on the total number of cycles, and on the longest cycle. 

Some interesting questions for future research include:
\begin{enumerate}
\item \textit{Non-i.i.d.\ or correlated payoffs.}
Our results rely on the i.i.d.\ structure of payoff entries.
Extending the analysis to payoffs with some general correlation structure,
as initiated by Mimun, Quattropani, and Scarsini~\cite{mimun2024best},
is an important direction. In particular, what is the cycle structure
when player payoffs are correlated with coefficient $\rho \in [-1,1]$?
\item \textit{More than two players.}
The best-response digraph formalism extends naturally to $m$-player games
as a directed graph with $n^m$ vertices. Characterizing the cycle structure
in this regime, even for $m=3$, is an interesting problem.
\item \textit{Convergence time.}
Our results determine the limiting structural results, but do not address how long the best-response dynamics take to reach a cycle
or a PNE from a random initial condition. The time of the walk on $\calD$ to its limit
cycle is an interesting quantity to study.
\end{enumerate}



\section*{Acknowledgments}
We thank Brendan McKay for a discussion with R.-R.~Z. in 2024, during which the two discovered that the symmetric pair count in a uniformly random
cycle follows Poisson asymptotically, which motivated Lemma~\ref{lem:overlap}.

Nothing in the text is directly from large language models (LLMs).
ChatGPT and Claude were used to
\begin{itemize}
    \item[-] help to edit some TikZ figures and tables;
    \item[-] generate Python simulation code of random games to verify the claimed asymptotics;
    \item[-] audit some proof arguments, check grammar, and proofread.
\end{itemize}

\printbibliography
\addcontentsline{toc}{section}{References}

\end{document}